\documentclass{amsart}

\usepackage[dvipsnames]{xcolor}
\usepackage{mathtools,mathrsfs}
\usepackage{dsfont}
\usepackage[makeroom]{cancel}
\usepackage{graphicx}

\usepackage{subcaption}

\usepackage{booktabs}
\usepackage{pgfplots}

\pgfplotsset{compat=newest} 

\usetikzlibrary{arrows.meta,decorations.markings}

\usepackage[colorlinks,citecolor=blue,urlcolor=blue]{hyperref}

\usepackage{enumerate}

\mathtoolsset{showonlyrefs,showmanualtags}

\usepackage{stackengine}
\newcommand\circledtriangle{%
  \mathbin{\stackinset{c}{}{c}{}{\resizebox{0.9\width}{!}{$\bigtriangleup$}}{\resizebox{\width}{!}{$\bigcirc$}}}%
}

\newcommand{\myx}{x}
\newcommand{\myy}{y}

\newcommand{\myX}{X}
\newcommand{\myY}{Y}

\newcommand{\jelclass}[1]{\thanks{\textit{JEL Classification.} #1}}

\newtheorem{theorem}{Theorem}[section]
\newtheorem{definition}[theorem]{Definition}
\newtheorem{proposition}[theorem]{Proposition}
\newtheorem{lemma}[theorem]{Lemma}
\newtheorem{corollary}[theorem]{Corollary}

\newtheorem{example}[theorem]{Example}
\newtheorem{remark}[theorem]{Remark}

\newtheorem{assumption}[theorem]{Assumption}
\newtheorem{notation}[theorem]{Notation}

\newtheorem*{proposition*}{Proposition}
\newtheorem*{theorem*}{Theorem}
\newtheorem*{criterion*}{Criterion}
\newtheorem{remark*}{Remark}

\DeclareMathOperator{\supp}{supp}

\definecolor{darkgreen}{rgb}{0.0, 0.5, 0.0}

\pgfplotsset{compat=1.18}

\usepackage[title]{appendix}

\begin{document}

\begin{abstract}
We present a formal framework for the aggregation of financial markets mediated by arbitrage. Our main tool is to characterize markets via utility functions and to employ a one-to-one correspondence to limit order book states. Inspired by the theory of thermodynamics, we argue that the arbitrage-mediated aggregation mechanism gives rise to a \emph{market-dynamical entropy}, which quantifies the loss of liquidity caused by aggregation. As a concrete guiding example, we illustrate our general approach with the Uniswap v2 automated market maker protocol used in decentralized cryptocurrency exchanges, which we characterize as a so-called \emph{ideal market}. We derive its equivalent limit order book representation and explicitly compute the arbitrage-mediated aggregation of two liquidity pools of the same asset pair with different marginal prices. We also discuss future directions of research in this emerging theory of market dynamics.
\end{abstract}

\title{Aggregation of financial markets}

\author{Georg Menz}

\address{Department of Mathematics, University of California, Los Angeles \newline Mathcraft LLC, Los Angeles}
\email{gmenz@math.ucla.edu, g@mathcraft.us}

\author{Moritz Voss}

\address{Department of Mathematics, University of California, Los Angeles}
\email{voss@math.ucla.edu}

\jelclass{C02, D53}

\subjclass[2020]{91G15, 91B50}

\date{\today}

\maketitle

\tableofcontents

\section{Introduction} \label{sec:introduction}

The theme of this article is to develop and study mathematical tools that are useful to describe stationary financial markets in equilibrium. There are various types of financial markets: limit order book exchanges, decentralized liquidity pools, dark pools, alternative trading systems, etc. Therefore, we strive for a fairly general definition of a market. When defining a market, there are two main perspectives: The first one is simplistic and reduces a market to a device or mechanism that allows to exchange one asset into another. Hence, in order to describe a market one would only have to characterize how much of one asset could be exchanged on the market into another asset. The second perspective understands a market as a place where traders meet and interact. Describing a market from this perspective is much more subtle, as one needs to understand and model how traders influence each other and come to an agreement. To meet both views, we use utility functions (see Section~\ref{s_markets}). The super-level set of the utility function describes the set of all possible trades with that market, fully characterizing the market as an exchange mechanism. The behavior of rational traders can also be described via utility functions. A trader is only willing to trade if her utility increases (or at least does not decrease). The utility function of the market can then also be interpreted as the aggregation of the utility functions of the individual traders present in the market. \\

The main goal of this article is to describe a natural aggregation mechanism in financial markets. The described mechanism is quite general. It applies whenever arbitrageurs are present in the market, assuming they are willing and able to close arising arbitrage opportunities. As real financial markets are highly fragmented (see, e.g.,~\cite{SEC2020}) aggregation is an important tool to describe the global market. Specifically, the mechanism allows aggregation of conventional markets like centralized limit order book exchanges with non-conventional markets like decentralized automated market makers. If done correctly, the aggregate market will reflect the preferences and opinions of all market participants, representing the true state of the global financial market. The National Best Bid Offer (NBBO) price execution requirement in equity markets practically forces brokers to act on the aggregated market on behalf of their clients (see Regulation National Market System \cite{SEC2005}). Smart order routing in order to minimize price impact also effectively trades on an aggregated market. \\

Deriving the utility function of the aggregated system from the utility functions of individual agents is a common theme in many areas of research, like in economics or social welfare theory. When aggregating financial markets, one faces several challenges. The first challenge is that the theory must be general enough to bridge from microeconomics to macroeconomics. For example, for an individual trader money might be used as an absolute measure of value, which is not possible anymore on a macroeconomic level. Another difficulty is that the process of aggregation of utility functions is highly nontrivial. The aggregation process models how different agents come to an agreement. Therefore, it is not surprising that there exist a multitude of aggregation mechanisms. For instance, in social welfare theory there are utilitarian, Nash bargaining, Bergson-Samuelson, and many more aggregation mechanisms. Ultimately, there is no ideal canonical mechanism which follows from Arrow's impossibility theorem~\cite{Arr:12}. The theorem states that it is impossible to construct a social welfare function that satisfies a small number of reasonable conditions like, e.g., unanimity, non-dictatorship, and independence of irrelevant alternatives. \\

Our situation differs substantially from Arrow's setting. On the one hand, we face a fairly complex aggregation problem as traders can choose among a continuum of choices and not between a finite number of alternatives. On the other hand, this complexity is compensated by the regularity of utility functions, which makes optimal aggregation easier to achieve. In the end, we pursue a different objective. Instead of trying to find and characterize an optimal theoretical aggregation mechanism, we just strive to identify and characterize a natural mechanism that reflects how financial markets are aggregated in real life. More precisely, the main contribution of our work is to identify a realistic arbitrage-mediated aggregation mechanism and to study its implications. As a byproduct, we also rigorously describe the equivalence between automated market makers and limit order books, which might be of independent interest. \\

To describe the aggregation mechanism, we make strong assumptions. We only consider stationary markets in equilibrium and assume that the market participants are transparent and isolated. Isolated means that the traders have no access to a market before interacting, and therefore cannot change their initial portfolio before aggregation. Transparent means that the traders communicate their intentions via limit orders to the public. Under this assumption, the arbitrage-mediated aggregation is Pareto-optimal: After aggregation, no trader can improve her utility without decreasing another trader's utility. \\

Transparency is obviously not satisfied in real financial markets. In reality, traders are non-transparent as there is no need to reveal their motives. For example, iceberg, fill-or-kill, and immediate-or-cancel orders are not fully visible to the public and can be used to hide intentions. There is even a game-theoretic incentive not to act transparent. If a trader wants to increase her utility more than it would be possible by transparent aggregation, she needs to hide her information. As an analogy, if a player wants to win in a highly competitive card game, she will not show her hand to her opponents. As a consequence, the financial market will never be transparent. We will briefly discuss some further aspects of the non-transparent case in Appendix~\ref{sec:hidden_markets}. \\

Let us sketch the main idea of the aggregation mechanism. Aggregation takes place on the level of (hypothetical) limit order books and relies on the following observations:
\begin{itemize}
    \item Utility functions are characterized up to a monotone transformation by their indifference curves. Therefore, it suffices to describe the indifference curve of the aggregated market (see Section~\ref{s_isoutils}).  
    \item Each participant in a market (trader) is described via a utility function that encodes the trader's preferences for different portfolios. It also characterizes all possible trades with the trader via the utility function's indifference curves (see Section~\ref{s_markets}). We call utility indifference curves \emph{iso-utils} because the utility of the trader's portfolio is constant on these curves. 
    \item It turns out that iso-utils are in one-to-one correspondence to hypothetical limit order book states (see Section~\ref{s_isoutil_LOB}).
    \item On the level of limit order books, aggregation is straightforward: the limit orders of the individual traders are combined into one joint book. 
\end{itemize}

After the aggregation, one faces the following problem: As traders have different opinions, the aggregated hypothetical limit order book might not be settled. This means that buy and sell side of the book might overlap. As we will discuss, it therefore cannot describe a convex iso-util as it is needed to characterize the utility function of the aggregated market. The way out is to settle (or clear) the hypothetical limit order book which effectively results in a natural convexification procedure. This process is mediated by arbitrageurs as the overlap opens arbitrage opportunities. Specifically, arbitrageurs will counter-trade the overlapping limit buy and sell orders, settle the hypothetical limit order book and make a risk-free profit. Such a settlement mechanism can also be observed in real financial markets, for example, in darkpools or via cross-exchange price arbitrage trading. The settled hypothetical limit order book then describes a convex iso-util which defines the utility function of the aggregated market. \\

The aggregation mechanism still has one degree of freedom, namely how traders react to the counter-trading of the arbitrageur. Therefore, we distinguish between two types of traders:
\begin{itemize}
\item An \emph{adiabatic} trader only trades if her utility is strictly increased. This type of trader is not willing to undo any trade.
\item An \emph{iso-util} trader is always willing to trade as soon as their utility is not decreased. This trader is willing to undo any trade which keeps her utility invariant. 
\end{itemize}

These two types of traders lead to two canonical settlement mechanisms:
\begin{itemize}
\item In an \emph{adiabatic settlement} only adiabatic traders are present. Hence, all overlapping limit orders vanish out of the limit order book after the arbitrage-mediated clearing process.
\item In an \emph{iso-util settlement} only iso-util traders are present. As a consequence, overlapping limit orders reappear on the other side of the limit order book after the arbitrage-mediated clearing process.
\end{itemize}

Aggregation in real financial markets is often comprised of a mixture of both settlement mechanisms, reflecting that both adiabatic and iso-util traders might be present in the market. Even if adiabatic aggregation seems to be more natural, iso-util aggregation plays an important role. For instance, protocols of automated market makers are the most prominent example. Indeed, in the Uniswap v2 and v3 protocol (see~\cite{Adams:20, Adams:21}), liquidity providers act iso-util. Moreover, their individual iso-utils are effectively aggregated in an iso-util fashion to determine how the protocol provides liquidity to liquidity takers according to an aggregated iso-util. \\

As individual traders are never in the exact same equilibrium, we have the following central observation:
\begin{center}
\emph{When markets aggregate some liquidity is lost to arbitrage.}
\end{center}

Inspired by thermodynamics, we call this observation the \emph{fundamental law of market dynamics}. The lost liquidity is called \emph{market-dynamical entropy}. These are the first building blocks of what we envision to become a new theory of \emph{market dynamics}. For more details about the inspiration and vision to develop this new theory, we refer to Section~\ref{sec:aggregate_markets} and~\ref{sec:market_dynamics}.\\

The best practical application of our aggregation method is when aggregating individual exchanges to one virtual global exchange. Our treatment is general enough to aggregate both discrete and continuous exchanges. For example, in crypto the trading pair BTC-USDT is traded on many different centralized and decentralized exchanges. Centralized limit order books are discrete markets, whereas decentralized exchanges like Uniswap are usually continuous markets. Aggregating those markets into one global BTC-USDT market gives a better picture of the current supply and demand relationship compared to just focusing on one exchange. \\

The aggregation mechanism described in this article is not restricted to financial markets. It applies whenever an arbitrageur is present and able to close arising arbitrage opportunities. In Appendix~\ref{sec:entropy_in_economics}, we describe how this type of market aggregation can be applied to consumer and producer markets such as the car market. In this setting, the car dealership is playing the role of the arbitrageur. Market-dynamical entropy represents how many cars are sold to customers and how much money got transferred to the producers and the car dealership. That opens a new role of market-dynamical entropy: as a measure of economic activity. In particular, arbitrage-mediated aggregation can be used to describe how microscopic agents aggregate into one global aggregated market; and how economic activity results from this aggregation process. \\

Let us briefly summarize the main contributions of this article:
\begin{itemize}
\item We show that utility functions provide a tool to describe stationary markets in equilibrium. They are complex enough to capture the subtle inter-dependencies of assets and non-linear effects like price impact. An explicit example are automated market maker protocols (see Sections~\ref{s_axiomatic_introduction_markets} and~\ref{s_isoutils}).
\item We discuss that utility functions can be described via iso-utils; and that iso-utils satisfy a one-to-one correspondence to (hypothetical) limit order book states. Our treatment is general enough to study discrete and continuous markets in one unifying framework (see Sections~\ref{s_isoutils} and~\ref{s_isoutil_LOB}).
\item We use this correspondence as a formal aggregation mechanism by which liquidity gets lost due to price arbitrage. Our arbitrage-mediated aggregation mechanism allows to aggregate financial markets from the atomistic level of individual traders to the global market. As a concrete example, we illustrate our approach with the Uniswap v2 protocol and explicitly compute the arbitrage-mediated aggregation of two different liquidity pools of the same asset pair with different marginal prices (see Section~\ref{sec:aggregate_markets}).     
\item We introduce and study new notions like ideal markets, the fundamental law of market dynamics, and the market-dynamical entropy. Additionally we discuss open problems and how to further develop the emerging theory of market dynamics (see Sections~\ref{sec:aggregate_markets} and~\ref{sec:market_dynamics}). \\
\end{itemize}

To the best of our knowledge, formally describing mechanisms for the aggregation of financial markets in a general, unifying framework seems to be barely studied in the literature, especially within our context of centralized limit order book exchanges and decentralized liquidity pools. Limit order books and automated market makers have been compared before, but not with the purpose of aggregation. We refer to~\cite{MilionisMoallemiRoughgarden:24} and the references therein. Other recent papers studying automated market makers include~\cite{CapponiJia:21, CarteaDrissiMonga:23a, CarteaDrissiMonga:23b, EchenimGobetMaurice:24}. Our general idea of describing markets via utility functions is inspired from~\cite{BichuchFeinstein:22} and our description of the equivalence of automated market makers and limit order books generalizes the analysis in~\cite{Young:20}. For
an extensive treatment of limit order books, we refer to~\cite{Bouchaud:18}.\\     

The remainder of the article is structure as follows: In Section~\ref{s_axiomatic_introduction_markets}, we axiomatically introduce markets via utility functions. In Section~\ref{s_isoutils}, we discuss the concept of iso-utils, also known as indifference curves, and how they characterize marginal prices and price impact. Section~\ref{s_isoutil_LOB} comprises the main mathematical work of this article. It deduces a one-to-one relationship between limit order books and iso-utils, and describes how to clear an unsettled (hypothetical) limit order book. In Section~\ref{sec:aggregate_markets}, we apply these concepts to describe the arbitrage-mediated aggregation of markets. We also discuss the fundamental law of market dynamics, market dynamical entropy and its applications. As a concrete example, we consider the adiabatic and iso-util aggregation of ideal markets. In Section~\ref{sec:market_dynamics}, we discuss the emerging theory of market dynamics and open problems. In Appendix~\ref{sec:entropy_in_economics}, we illustrate the aggregation of consumer and producer markets. Appendix~\ref{sec:hidden_markets}  discusses some aspects of aggregation of non-transparent markets. \\

\section{Axiomatic description of markets}\label{s_axiomatic_introduction_markets}  \label{s_markets}

In this section we review the mathematical framework to describe markets. Our approach is straightforward. We describe the financial market as a pure exchange economy of agents trying to improve the utility of their endowed portfolio through trading. Let us look at one individual participant of a market subsequently called trader~$l$. Trader~$l$ owns a portfolio of assets~$(X_1, \ldots, X_n)$ given by a vector~$x = (x_1, \ldots, x_n) \in [0, \infty)^n$, where $x_i$ denotes the number of units of the $i$-th asset held by the trader. Before aggregation the trader is isolated, which means that she cannot change her initial portfolio~$x$, for example by interacting with a market. \\

Obviously, trader~$l$ prefers certain portfolios over others. Those individual preferences can be encoded with the help of a utility function~$U_l : [0, \infty)^n \to \mathbb{R}$ which assigns to any portfolio~$x$ a utility~$U_l(x)$. For a precise definition of utility functions we refer to Definition~\ref{d_markets} below. As a rational trader, $l$ will only trade if it increases (or at least not decreases) her present utility. More precisely, trader~$l$ is willing to trade the portfolio~$x$ for the portfolio~$\tilde x$ if $U_l(x)\leq U_l(\tilde x)$. As a consequence the utility function~$U_l$ does not only encode the preferences of  trader~$l$ but also encodes the set of all possible trades~$\Delta x = \tilde x - x$ with trader~$l$ via the super-level set
\begin{align}
\left\{ \tilde x \in [0, \infty)^n \ | \ U_l(\tilde x ) \geq U_l (x) \right\}.
\end{align}
In classical micro-economics, the trader would interact with a market and optimize her utility given a budget constraint which is imposed by her initial portfolio~$x$. We do not follow this path. In a nutshell, we propose the following alternative for financial markets where arbitrageurs are present: The arbitrageurs take advantage of the information asymmetry and match traders with disagreeing price expectations making a profit. The gain of utility that can be achieved by the interaction between traders is therefore transferred from the individual trader to the arbitrageurs. \\

Let us now turn to markets. What is a market? Oversimplifying, a market is just a device or mechanism that allows one to exchange one good for another. From this perspective, even a single trader~$l$ can be regarded as a market. All possible trades can be derived from her utility function~$U_l$, and the trader's present portfolio~$x$ describes the maximal quantities one can purchase from this atomic market. \\

In common language, a market denotes a place where traders meet and trade goods. Therefore, let us now assume that a couple of individual traders enumerated by~$1$ to~$k$ meet at a market place. It is a central question of economics how the traders, each equipped with individual preferences, come to an overall agreement. Mathematically, this amounts to the question of how the individual utility functions $(U_1, \ldots, U_k)$ are aggregated into one joint utility function~$U = U_1 \bigtriangleup \cdots \bigtriangleup U_k$ of the market. In Section~\ref{sec:aggregate_markets} below, we propose a simple arbitrage-mediated aggregation mechanism. In this section, let us assume that the traders reached agreement. Then the market comprised of the traders~$1, \ldots, k$ can again be described by a utility function~$U$ and a portfolio~$x$. The portfolio~$x$ describes the maximal quantities that are supplied on the market and will be called \emph{supply level} further on. The set of possible trades with this market is again characterized via the super-level set of the utility function
\begin{align}
\left\{ \tilde x \in [0, \infty)^n \ | \ U(\tilde x ) \geq U (x) \right\}.
\end{align}
Let us iterate the meaning of utility~$U(x)$ one more time: It represents the value and preference the market puts on the portfolio~$x$ of assets.

\begin{assumption}
In this article, we only consider transparent stationary markets. This means that the utility functions of traders and markets do not change over time and are observable. We also assume that the traders or markets that are being aggregated are transparent and isolated, i.e., the utility functions and portfolios of each trader/market are observable and the traders/markets cannot interact before aggregation.
\end{assumption}

Let us now turn to the precise definition of a market.

\begin{definition}[Markets, supply levels and utility functions] \label{d_markets}
A market~${\mathcal{M}}$ of the assets ~$(\myX_1, \ldots, \myX_n)$ is defined via supply levels~$(\myx_1, \ldots, \myx_n) \in [0, \infty)^n$ and a utility function~$U:[0, \infty)^n  \to \mathbb{R} \cup \left\{ - \infty \right\}$. The utility function assigns to every supply level~$(\myx_1, \ldots, \myx_n)$ a utility~$U(\myx_1, \ldots ,  \myx_n)$. It has to satisfy the following conditions:
\begin{enumerate}
\item {\bf(Continuity)} The function~$U$ is continuous.
\item {\bf (Quasi concavity)} The function~$U$ is quasi-concave, i.e., for every~$T >0$ the sets
\begin{align}
\left\{ (\myx_1, \dots, \myx_n) \in [0,\infty)^n \ | \  U(\myx_1, \ldots,  \myx_n) \geq  \log T \right\}        
\end{align}
are convex.
\item {(\bf Strict monotonicity)} For any asset~$i \in \left\{1, \ldots, n \right\}$ and supply levels $(\myx_1, \ldots, \myx_n )$ the function~$\myx_i \mapsto U(\myx_1, \dots, \myx_i, \ldots, \myx_n)$ is strict monotone increasing.
\item {\bf (Unbounded from above)} For any asset~$i \in \left\{1, \ldots, n \right\}$ it holds \newline $\lim_{\myx_i \to \infty} U(\myx_1, \ldots, \myx_n) = +\infty$ . 
\end{enumerate}
\end{definition} 

The utility function determines many aspects of a market:
\begin{itemize}
\item The set of possible trades is described via super-level sets of the utility function;
\item The convexity condition encodes the non-existence of arbitrage opportunities (see Section~\ref{s_isoutil_LOB});    \item The marginal and equilibrium prices are a function of the utility landscape (see Section~\ref{s_isoutils});
\item The price impact is a function of the utility landscape (see Section~\ref{s_isoutils}).
\end{itemize}

\begin{remark}
    To keep the presentation simple and accessible we restrict ourselves to markets with only two assets~$(\myX, \myY)$. The supply levels are denoted by~$(\myx, \myy)$. 
\end{remark}

In this manuscript we focus on financial markets, though Definition~\ref{d_markets} also describes non-financial markets. To have a concrete example of an asset pair in mind one could think of the asset pair~$(\myX, \myY)$ as~$(\textrm{US dollar, gold})$ where the unit of gold is measured in ounces.

\begin{definition}[Ideal market] \label{def:ideal_market} A market associated to the utility function~$U(x,y)= \log x + \log y$ is called~\emph{ideal market} (see Figure~\ref{fig:ideal_market_utility_function}). For the origin of this terminology see~Section~\ref{sec:market_dynamics}.
\end{definition}

\begin{figure}[ht]
\centering
\begin{tabular}{@{}c@{}c}
\includegraphics[height=2in]{ideal_utility.png} &
\begin{tikzpicture}[scale=0.75]
\begin{axis}[
xlabel= Supply level $x$,
ylabel= Supply level $y$,
zlabel style={rotate=-90},
]
\addplot3 [
surf,
domain=0.01:5,
domain y=0.01:5,
] {ln(x) + ln(y)};
\end{axis}
\end{tikzpicture}
\end{tabular}
\caption{Contour plot (left) and 3D-plot (right) of the utility function $U(x,y)=\log x + \log y$.}
\label{fig:ideal_market_utility_function}
\end{figure} 

\begin{example}
The decentralized cryptocurrency exchange protocol Uniswap v2 uses the ideal market for its automated market maker protocol (see~\cite{Adams:20}). The successor protocol Uniswap v3 allows any utility function and therefore is able to characterize the true supply and demand relationship in the market (see~\cite{Adams:21}).
\end{example}

\begin{example}
Another example of a utility function is the Cobb-Douglas production function~$U(x,y)=Ax^\beta y^\alpha$ for some~$A>0$ and~$0< \alpha, \beta< 1$.
\end{example}

\begin{remark}[Discrete supply levels] \label{r_discrete_supply_levels} 
Often supply levels must be discrete as it is not possible to buy fractions of certain assets. Definition~\ref{d_markets} can be extended to discrete supply levels by introducing the set of admissible supply states as a discrete subset~$\mathcal{P} \subset [0, \infty) \times [0, \infty)$. In this scenario, the utility function~$U$ would still be defined on a continuum but supply levels must take values in~$\mathcal{P}$.
\end{remark}

\begin{remark}[Utility functions of real markets]\label{rem:d_explicit_markets}
In real markets utility functions are often not observable. To address this complexity, we propose the following classification of markets: 
\begin{itemize}
\item In \emph{implicit} markets one cannot directly observe the utility function. However, as we explain later, it is possible to deduce some information about the utility function from observing trades, prices, volume and activity of the market. An example would be Over-the-Counter (OTC) markets.
\item In \emph{semi-explicit} markets only certain features of the utility function can be observed. An example are centralized Limit Order Book (LOB) exchanges like stock exchanges or foreign exchange markets. It is possible to read off the present so-called iso-utility from the current limit order book state (see Section~\ref{s_isoutil_LOB}).
\item In \emph{explicit} markets the utility function of the market is explicitly given. An example are Automated Market Makers (AMMs). For instance, the original Uniswap protocol uses the utility function~$U(\myx, \myy) = \log(\myx) + \log (\myy)$. For more information and further examples we refer to~\cite{BichuchFeinstein:22}.
\end{itemize}
\end{remark}

\begin{remark} [Interpretation of the supply level~$(\myx, \myy)$] Tied to a market~$\mathcal{M}$ is the notion of supply levels~$(\myx, \myy)$ of the asset pair $(\myX, \myY)$ which has the following interpretation:
\begin{itemize}
\item The number~$\myy$ denotes the maximal number of units of the asset~$\myY$ that can be traded for the asset~$\myX$. The number $\myx$ denotes the maximal number of units of the asset~$\myX$ that can be traded for the asset~$\myY$ (see Section~\ref{s_isoutil_LOB}).
\item The level~$(\myx, \myy)$ may represent the equilibrium point of the supply and demand curve of the asset pair~$(\myX, \myY)$.
\end{itemize}
\end{remark}

\begin{remark}[Money as an absolute versus relative measure of value] 
In traditional economic theory, the role of money is paramount in establishing the interplay of supply and demand. Often, money is employed as an absolute measure of value in the sense that one dollar always has the exact same value. This works well in microeconomics because of the relatively small quantities involved. The use of money as an absolute measure of value becomes problematic when attempting to bridge the gap between micro- and macroeconomics. When developing a theory of markets that can seamlessly transition from micro- to macroeconomic markets via aggregation, it becomes necessary to employ money as a relative measure of value, i.e., the value of one US dollar is relative and depends on many factors; e.g., on the overall money supply or on the overall supply of consumer goods. For this reason, our framework assigns fiat money (US dollar), usually denoted as~$X$, the same role as any other asset or good, usually denoted as~$Y$. The value of money is derived from its use, namely that it can be easily exchanged for any other asset or good on a market~$(X,Y_1, \ldots, Y_n)$. 
\end{remark}

We note that the conditions in Definition~\ref{d_markets} have some important structural consequences on the utility function~$U$.

\begin{lemma}[Regularity of the utility function~$U$]\label{p_regularity_U}
Consider a market~$\mathcal{M}$ with utility function~$U$. Then the utility function~$U$ is almost everywhere differentiable with 
\begin{align}\label{e_sign_derivatives_U}
\infty >  \partial_\myx U >0 \qquad \mbox{and} \qquad  \infty >\partial_\myy U >0 \qquad \mbox{for a.e.~$(\myx, \myy)$.}
\end{align}
\end{lemma}

\begin{proof}
It is a well-known fact that quasi-convex functions are differentiable almost everywhere (see, e.g., comment after Theorem~5.3 in~\cite{CrZa:05}). This implies the existence of the partial derivatives~$\partial_\myx U$ and~$\partial_\myy U$. The lower bound in~\eqref{e_sign_derivatives_U} follows directly from the fact that a utility function~$U$ is strict monotone increasing by Definition~\ref{d_markets}. 
\end{proof}

\begin{remark}[Origins of the Definition~\ref{d_markets}]
In~\cite{BichuchFeinstein:22} a similar approach was used to describe automated market makers via utility functions. In this work, we use utility functions to describe general markets. Our assumptions on the utility function~$U$ given in Definition~\ref{d_markets} are less restrictive compared to~\cite{BichuchFeinstein:22}. Specifically, we excluded condition
\begin{align}
\text{\bf (Unbounded from below)} \qquad \lim_{\myx \to 0} U(\myx, \myy) = \lim_{\myy \to 0} U(\myx, \myy) = - \infty. 
\end{align}
Otherwise, Definition~\ref{d_markets} would not include a very important example, namely limit order book markets (see Section~\ref{s_isoutil_LOB}). Condition~{\bf (Unbounded from below)} forces level sets -- called indifference curves or iso-utils in our manuscript -- to always have an unbounded domain~$(0, \infty)$. However, when considering an iso-util coming from a (finite) limit order book, the domain of the iso-util is bounded. For more details we refer to Section~\ref{s_isoutils}, Proposition~\ref{p_iso_util_function}, and Section~\ref{s_isoutil_LOB}.
\end{remark}

\begin{remark}[Conditions in Definition~\ref{d_markets}]\label{r_conditions_markets}
The condition {\bf (Quasi concavity)} plays an important role. It means that the traders who constitute a market reached consensus, and hence rules out price arbitrage opportunities (cf.~Remark~\ref{rem:quasi_concave_no_arbitrage} below). If {\bf (Quasi concavity)} is not satisfied we say that the market is not settled. Under relatively weak regularity assumptions, quasi concavity can be strengthened to concavity, after a smooth, strict monotone reordering of preferences (see Theorem~3 in~\cite{ConRas:15}).  The condition {\bf (Strict monotonicity)} implies that more supply always implies more utility. {\bf (Unboundedness from above)} corresponds to unsaturability.  The combination of {\bf (Quasi concavity)} and {\bf (Strict monotonicity)} entails that on a market with higher utility more assets can be exchanged with less price impact. If one would use {\bf (Strict concavity)} instead of~{\bf (Quasi concavity)} then distributed portfolios would have higher utility than concentrated portfolios.
\end{remark}

\begin{remark}[Additional conditions on the utility function~$U$] There are several additional conditions that can be added to the definition of a utility function, as for example Inada, Inada+ or single-crossing conditions. Those conditions serve several purposes; for instance, to ensure existence of solutions to optimization problems or path independence of a trading loop. We refer to ~\cite{BichuchFeinstein:22} for more details.
\end{remark}


\section{Temperature, iso-utils and marginal prices}\label{s_isoutils}

In this section we review how (marginal) prices are calculated in markets. 
\begin{definition}[Temperature and mean activity] \label{d_def_temperature}
We consider a market~$\mathcal{M}$ with utility function~$U$ at supply level~$(\myx, \myy)$. Then we call
\begin{align}
T: = e^{U(\myx, \myy)}
\end{align}
the temperature of the market~$\mathcal{M}$ at supply level~$(\myx, \myy)$. We call~$A := \sqrt{T}$ the mean activity of the market~$\mathcal{M}$ at supply level~$(\myx, \myy)$.
\end{definition}

\begin{remark}
We call~$T$ temperature because of an analogy to thermodynamics (cf.~Example~\ref{rem:isoutils} below). Moreover, we call~$\sqrt{T}$ mean activity because it coincides with the mean arrival rate of buyers and sellers in a canonical microscopic model of a market, modeling the arrivals of buyers and sellers via Poisson processes. 
\end{remark}

Definition~\ref{d_def_temperature} of temperature motivates the definition of iso-utils which identify supply levels~$(\myx, \myy)$ with the same temperature. In microeconomics, iso-utils are known under the name~\emph{indifference curves}~(see, e.g.,~\cite{HicAll:34}). We prefer to use the name iso-util to point out similarities to iso-thermals in thermodynamics. 

\begin{definition}[Iso-utils]
We consider a market~$\mathcal{M}$ with utility function~$U$. The iso-utils of~$\mathcal{M}$ are defined as the level sets of the utility function~$U$, i.e., the iso-util~$I_T$ of~$U$ associated to the temperature~$T \in \mathbb{R}_+$ is defined by
\begin{align}
I_T:= \left\{ (\myx, \myy) \ | \  U (\myx, \myy) = \log T \right\}.
\end{align}
With slight abuse of notation we suppress index~$T$ and write $I$ instead of~$I_T$.
\end{definition}

\begin{figure}
\centering
\includegraphics[scale=0.25]{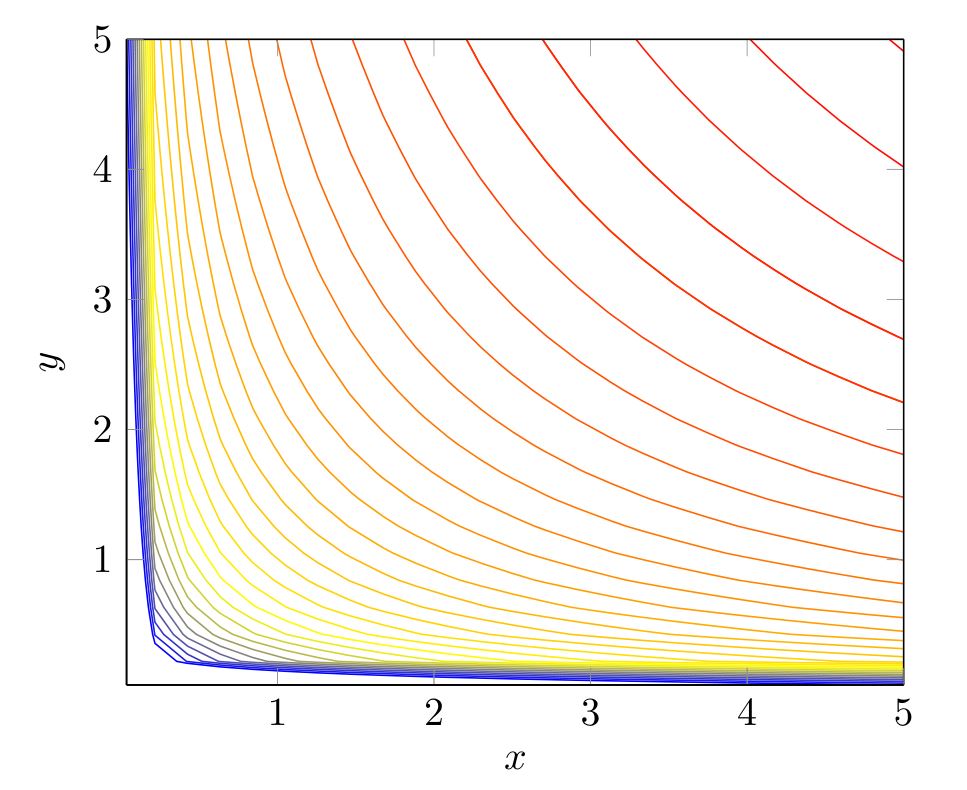}
\caption{Iso-utils of an ideal market with utility function~$U(x,y)=\log x + \log y$}
\label{fig:ideal_market_iso_utils}
\end{figure}

\begin{example}[Graphs of Iso-utils] \label{rem:isoutils}
In Figure~\ref{fig:ideal_market_iso_utils}, we plot the iso-utils of an ideal market (see Definition~\ref{def:ideal_market}). They are given by the formula $x \cdot y = T$, which we call \emph{ideal market law} in analogy to the ideal gas law of thermodynamics. The ideal gas law is given by the formula~$P \cdot V = c \cdot T$, where~$P$, $V$, and $T$ denote pressure, volume, and temperature, respectively, of the underlying thermodynamic system. In Figure~\ref{fig:LOB_isoutil_settled} below, we illustrate the iso-util of a limit order book.
Figure~\ref{fig:LOB_BTCUSDT_isoutil} illustrates an iso-util that is associated to the Binance trading pair BTC-USDT at a particular snapshot in time. We refer to Example~\ref{ex:binance_BTC_USDT} for more details.
\end{example}

\begin{remark}[Role of iso-utils]
Iso-utils play a central role in our theory of market dynamics as they characterize utility functions up to strict monotone transformations. Together with the current supply level, iso-utils determine the best possible trades in a market, and therefore (marginal) prices as well as price impact; see Definition~\ref{d_prices} below.
\end{remark}

Iso-utils can always be described as a graph of a function.

\begin{proposition}[Function representation of an iso-util]\label{p_iso_util_function}
Consider a market~$\mathcal{M}$ with utility function~$U$. For arbitrary~$T\in (0, \infty)$ we consider a non-empty iso-util
\begin{align}
I = \left\{ (\myx, \myy) \ | \  U (\myx, \myy) = \log T \right\} \neq \emptyset.
\end{align}
Then there exist a number~$d_f \in \mathbb{R}_+ \cup \left\{ \infty\right\}$ and a function~$f: D_f \to [0, \infty)$ defined on $D_f = (0, d_f)$ such that:
\begin{enumerate}
\item[(i)] The function~$f$ is convex.
\item[(ii)] It holds~$\lim_{x\to d_f} f(x) =0$.
\item[(iii)] The iso-util~$I$ is the graph of~$f$, i.e.,
\begin{align}
I = \left\{ (\myx, f(\myx)) \ | \ \myx \in D_f \right\}.
\end{align}
\end{enumerate}
The function~$f$ is called function representation of the iso-util~$I = \left\{ U = \log T \right\}$. Moreover, the left and right derivatives~$f_{-}'$ and~$f_{+}'$ of the function $f$ exist everywhere on~$D_f$ and $f$ is differentiable almost everywhere (a.e.) on~$D_f$. It also holds that
\begin{align}\label{eq:derivative_of_function_representation}
f'(x) =   -\frac{\partial_\myx U (\myx, f(\myx))}{\partial_\myy U (\myx, f(\myx))} \qquad \text{a.e.}
\end{align}
The functions~$f_{-}'$, $f'$, and $f_{+}'$ are non-decreasing and satisfy     
\begin{align}\label{eq:ordering_left_right_hand_derivative1}
- \infty< f_{-}'  \leq f_{+}'< 0
\end{align}
everywhere and
\begin{align}\label{eq:ordering_left_right_hand_derivative2}
- \infty<  f_{-}'  \leq  f' \leq f_{+}' < 0 \qquad \text{a.e.}
\end{align}
If the utility function~$U$ is unbounded from below then the domain of the function~$f$ is given by~$D_f =(0, \infty)$.
\end{proposition}

\begin{proof}[Proof of Proposition~\ref{p_iso_util_function}]
The aim is to define a function~$f: (0, d_f) \to [0, \infty)$ which satisfies the desired conditions. We start with observing that by definition of quasi concavity the iso-util~$\left\{  U = \log  T \right\}$ is the boundary of the convex set~$\left\{   U \geq  \log T \right\}$.
Let us consider the set $D_f$ as given by the projection of~$I$ onto the~$\myx$-axis. Because the set~$I$ is convex it follows that~$D_f = (0, d_f)$ for some~$d_f \in \mathbb{R} \cap \left\{ \infty \right\}$. Both observations imply that the normalized tangent vector~$\vec{v}(\myx, \myy)$ of the iso-util at point~$(\myx, \myy)$ exists almost everywhere and is orthogonal to the gradient~$\nabla U (\myx, \myy)$. We thus have
\begin{align}\label{eq:tangent_of_isoutil}
\vec{v}(\myx, \myy) = \frac{1}{|\nabla U(\myx, \myy)|} \left( \partial_\myy U(\myx, \myy) ,  - \partial_\myx U (\myx, \myy)\right)^\top.
\end{align}
By~\eqref{e_sign_derivatives_U} it holds~$\partial_\myy U>0$ and hence~$\vec{v}$ can be written as
\begin{align}
\vec{v}(\myx, \myy) = \frac{1}{\left|\left( 1 ,  -\frac{\partial_\myx U}{\partial_\myy U} \right)^\top\right|} \left( 1 ,  -\frac{\partial_\myx U}{\partial_\myy U} \right)^\top.
\end{align}
Moreover, we also obtain by~\eqref{e_sign_derivatives_U} that
\begin{align}\label{eq:sign_derivative_raw}
- \infty<  -\frac{\partial_\myx U}{\partial_\myy U} <0,
\end{align}
which implies that the tangent vector~$\vec{v} (\myx, \myy)$ always points to the lower right quadrant. Therefore, for any~$\myx_0 \in D_f$ the iso-util~$I$ contains exactly one element~$(\myx_0, \myy) \in I$ with~$\myx$-coordinate given by~$\myx_0$. This is enough to show that the iso-util~$I$ can be written as the graph of a function~$f$. 
    
Let us now verify the claimed properties of the function~$f$. First, it must hold that~$\lim_{x\to d_x} f(x) =0$, otherwise we would have a contradiction to the fact that the utility function~$U$ is unbounded from above. Next, we observe that the epigraph of the function~$f$ is given by the super-level set~$\left\{(x,y) \ | \ U(x,y) \geq \log T \right\}$, which is convex by the quasi concavity of the utility function~$U$. Therefore, the function~$f$ is convex; and this implies existence and monotonicity of~$f_{-}'$,~$f_{+}'$ everywhere; as well as the existence and monotonicity of~$f'$ a.e. The desired identity~\eqref{eq:derivative_of_function_representation} follows from~\eqref{eq:tangent_of_isoutil} and the desired inequality~\eqref{eq:ordering_left_right_hand_derivative1} follows from a combination of~\eqref{eq:derivative_of_function_representation} and~\eqref{eq:sign_derivative_raw}; as well as from the convexity of the function~$f$. In addition, the outer estimates in~\eqref{eq:ordering_left_right_hand_derivative2} follow from~\eqref{eq:ordering_left_right_hand_derivative1} and the observation that the function~$f$ is differentiable a.e; the inner inequality in~\eqref{eq:ordering_left_right_hand_derivative2} follows again directly from the convexity of the function~$f$. Finally, it  remains to show that~$D_f = (0, \infty)$ if the utility function is unbounded from below. To achieve this it suffices to show that for any~$x >0$ it holds~$x \in D_f$. Since the utility function~$U$ is unbounded from below, i.e., $\lim_{y \to 0} U(x,y) = - \infty$, unbounded from above, i.e., $\lim_{y \to \infty} U(x,y) = \infty$, and continuous it follows that there exists a~$y > 0$ such that~$U(x,y)= e^T$. But this implies~$x \in D_f$.  
\end{proof}

\begin{remark}
The same statement as in Proposition~\ref{p_iso_util_function} is true with the roles of~$\myx$ and~$\myy$ interchanged. Specifically, the iso-util~$I$ can also be written as the graph of a function~$g$ via $I = \left\{ (g(\myy), \myy) \ | \ \myy \in D_g \right\}$.
\end{remark}

\begin{remark}\label{rem:charcaterization_function_representation}
Let~$d_f \in \mathbb{R}_+ \cup \left\{ \infty \right\}$. We observe that as soon as a function~$f: (0, d_f) \to [0, \infty]$ is convex and satisfies~$\lim_{x\to d_f} f(x) =0$, then the graph~$I:=\left\{(x, f(x)) \ | \ x \in (0, d_f) \right\}$ defines an iso-util of some utility function~$U$.
\end{remark}

We distinguish between two parts of the iso-util, the ask part and the bid part. The reason for this becomes apparent in Section~\ref{s_isoutil_LOB} below.

\begin{definition}\label{def:bid_ask_part_isoutil}
Let us consider a non-empty iso-util
\begin{align}
I = \left\{ (\myx, \myy) \ | \  U (\myx, \myy) = \log T \right\} \neq \emptyset.
\end{align}
Assume that the present supply level of the market is given by~$(x_0,y_0)$. Then the set
\begin{align}
I_a : =  \left\{ (\myx, \myy) \in I \  \ | \  x \leq x_0 \right\}
\end{align}
is called the bid part of the iso-util~$I$, and the set
\begin{align}
I_b : =  \left\{ (\myx, \myy) \in I \  \ | \  x \geq x_0 \right\}
\end{align}
is called the ask part of the iso-util~$I$. 
\end{definition}

For an illustration we refer to Figure~\ref{fig:LOB_isoutil_settled}. Therein, the bid part of the iso-util is colored in blue and the ask part is colored in red. 

\begin{example}[Iso-utils of an explicit market]
As described in~\cite{BichuchFeinstein:22}, Uniswap v2 protocol iso-utils are given by the equations~$\myx \cdot \myy = T$ for all~$T \in \mathbb{R}_+$. They are illustrated in Figure~\ref{fig:ideal_market_iso_utils}.
\end{example}

\begin{example}[Iso-util of a limit order book market]
A limit order book can be understood as defining an iso-util of the underlying utility function and vice versa. For more details, we refer to Section~\ref{s_isoutil_LOB}.
\end{example}

\begin{remark}[Iso-utils and automated market makers]~\label{r_trading_automated_market_makers}
As mentioned above, an example of an explicit market are automated market maker liquidity pools. These are protocols that facilitate trading on blockchains. Simplifying, they work the following way. The protocol provides two pools: The first one is filled with~$\myx_0$ units of the asset~$\myX$ and the second with~$\myy_0$ units of the asset~$\myY$. Hence, the current supply level of the market is given by~$(\myx_0, \myy_0)$. When a trader wants to exchange~$\myx$ units of~$\myX$ into the asset~$\myY$ via the protocol, the protocol will take the~$\myx$ units of the trader and add them to the first urn. Then, it determines via a formula how many units~$\myy$ of~$\myY$ the trader receives out of the second pool. After the exchange there are~$\myx_0+ \myx$ many units of~$\myX$ in the first pool, and~$\myy_0 -\myy$ many units of~$\myY$ in the second pool. Therefore, the new supply level of this market is given by~$(x_0 + x, y_0 - y)$. How many units~$\myy$ the trader receives from the protocol is usually calculated via the help of an iso-util, which guarantees that after the trade the overall utility does not decrease. As described in~\cite{BichuchFeinstein:22}, this is achieved via the formula
\begin{align*}
\myy = \max \left\{ \myy \ | \ U(\myx_0 + \myx, \myy_0 - \myy) \geq  U(\myx_0, \myy_0)  \right\}.
\end{align*}
If the utility function is sufficiently regular, the maximum~$\myy$ is obtained on an iso-util, i.e., $U(\myx_0 + \myx, \myy_0 - \myy) =  U(\myx_0, \myy_0)$, or, equivalently,  $\myy = f(\myx_0) - f(\myx_0 + \myx)$ in terms of the iso-util's function representation $f$ from Proposition~\ref{p_iso_util_function}.
\end{remark}

Let us now explain how iso-utils determine prices in a market.

\begin{definition}[Realized price, marginal price, and price impact of asset~$\myY$]\label{d_prices}

Assume that in a market~$\Delta \myx>0$ units of~$\myX$ got exchanged into~$\Delta \myy>0$ units of~$\myY$. Then the realized price~$ P_{\frac{\myx}{\myy}}$ of one unit of~$\myY$ in terms of units of~$\myX$ in this trade is given by
\begin{align}
P_{\frac{\myx}{\myy}}(\Delta \myx ) = \frac{\Delta \myx}{ \Delta \myy}.
\end{align}
The realized price~$P_{\frac{\myy}{\myx}}$ in terms of units of~$\myY$ in this trade is given by
\begin{align}
P_{\frac{\myy}{\myx}}(\Delta \myx ) = P_{\frac{\myx}{\myy}}(\Delta \myx )^{-1} = \frac{\Delta \myy}{ \Delta \myx}.
\end{align}
If this trade caused an iso-util supply change, i.e., the supply level~$(\myx, \myy)$ before the trade and the supply level~$(\myx + \Delta \myx, \myy - \Delta \myy )$ after the trade are on the same iso-util~$I=\{U = \log T\}$, then~$P_{\frac{\myy}{\myx}}(\Delta \myx )$ is given by
\begin{align}
P_{\frac{\myy}{\myx}}(\Delta \myx ) =  \frac{f(x) - f(x + \Delta x)}{\Delta x}
\end{align}
where~$f$ denotes the function representation of the iso-util~$I$ (see Proposition~\ref{p_iso_util_function}). Sending the trading amount~$\Delta x \downarrow 0$ yields the marginal price. More precisely, we call
\begin{align}
P_{\frac{\myy}{\myx}} : = \lim_{\Delta \myx \downarrow 0 } P_{\frac{\myy}{\myx}}(\Delta \myx ) =  - f_{+}'(\myx) 
\end{align}
and
\begin{align}
P_{\frac{\myx}{\myy}} : = - \frac{1}{ f_{+}'(\myx)}
\end{align}
the marginal prices of~$\myY$ at supply level~$(\myx, \myy)$. Here,~$f_{+}'(x)$ denotes the right derivative of the function representation~$f$, which exists everywhere. The difference between the realized price and the marginal price is called price impact.     
\end{definition}

In economics, marginal prices are known by the name marginal rate of substitution.

\begin{remark}[Interpretation of the marginal prices~$P_{\frac{\myx}{\myy}}$ and $P_{\frac{\myy}{\myx}}$]
Consider a market at supply level~$(\myx, \myy)$. Then:
\begin{itemize}
\item $P_{\frac{\myx}{\myy}}$ expresses how many units of~$\myX$ one has to pay for one unit of~$\myY$ if the trade would be infinitesimal small.
\item $P_{\frac{\myy}{\myx}}$ expresses how many units of~$\myY$ one gets for one unit of~$\myX$ if the trade would be infinitesimal small.
\end{itemize}
\end{remark}

\begin{remark}[Realized price, marginal price, and price impact of asset~$\myX$]
Similar to Definition~\ref{d_prices} one can define the realized price, marginal price, and price impact of the asset~$\myX$. Then, the marginal price of the asset~$\myX$ is given by 
\begin{align}
\hat P_{\frac{y}{x}} = - f_{-}'(x) \quad \mbox{and} \quad 
\hat P_{\frac{\myx}{\myy}}  = - \frac{1}{ f_{-}'(\myx)}.
\end{align}
Here,~$f_{-}'$ denotes the left derivative of~$f$. The marginal price of~$X$ has following interpretation:
\begin{itemize}
\item $\hat P_{\frac{\myx}{\myy}}$ expresses how many units of~$\myX$ one receives for selling one unit of~$\myY$ if the trade would be infinitesimal small.
\item $\hat P_{\frac{\myy}{\myx}}$ expresses how many units of~$\myY$ one needs to sell in order to receive one unit of~$\myX$ if the trade would be infinitesimal small.
\end{itemize}
\end{remark}

\begin{remark}[Choice of num\'eraire]\label{rem:numeraire}
When quoting prices one needs to choose a num\'eraire, i.e., the standard by which value is measured. In the US stock market the num\'eraire is the US dollar, as prices of shares of stocks are expressed as an equivalent amount of US dollars. Note that the difference between~$P_{\frac{x}{y}}$ and~$P_{\frac{y}{x}}$ is the choice of the num\'eraire. For $P_{\frac{x}{y}}$, the num\'eraire is the asset~$X$ (which we usually think of as representing US dollar), and for~$P_{\frac{y}{x}}$ the num\'eraire is the asset~$Y$. 
\end{remark}

\begin{example}
Figure~\ref{fig:iso_util_prices} illustrates marginal and realized prices for a smooth iso-util of an ideal market (cf.~Definition~\ref{def:ideal_market}) of the form $x \cdot y = 37500$. Note that the marginal price of $Y$ in terms of units of $\myX$ is given by the inverse of the absolute value of the slope of the tangent line (green) through the current supply level~$(x,y)$, i.e., $P_{\frac{\myx}{\myy}} = - \frac{1}{ f_{+}'(\myx)} = \frac{75}{500} = 0.15$ with $f(x) = \frac{37500}{x}$. In contrast, the realized price for exchanging $\Delta x >0$ units of $X$ into $\Delta y > 0$ units of $Y$ is given by the inverse of the absolute value of the slope of the secant line (red) through the supply levels~$(x,y)$ and~$(x+\Delta x,y - \Delta y)$, that is, $P_{\frac{\myx}{\myy}}(\Delta \myx ) = \frac{\Delta \myx}{ \Delta \myy} = \frac{75}{250} = 0.3$.
\end{example}

\begin{figure}
\centering
\includegraphics[scale=0.65]{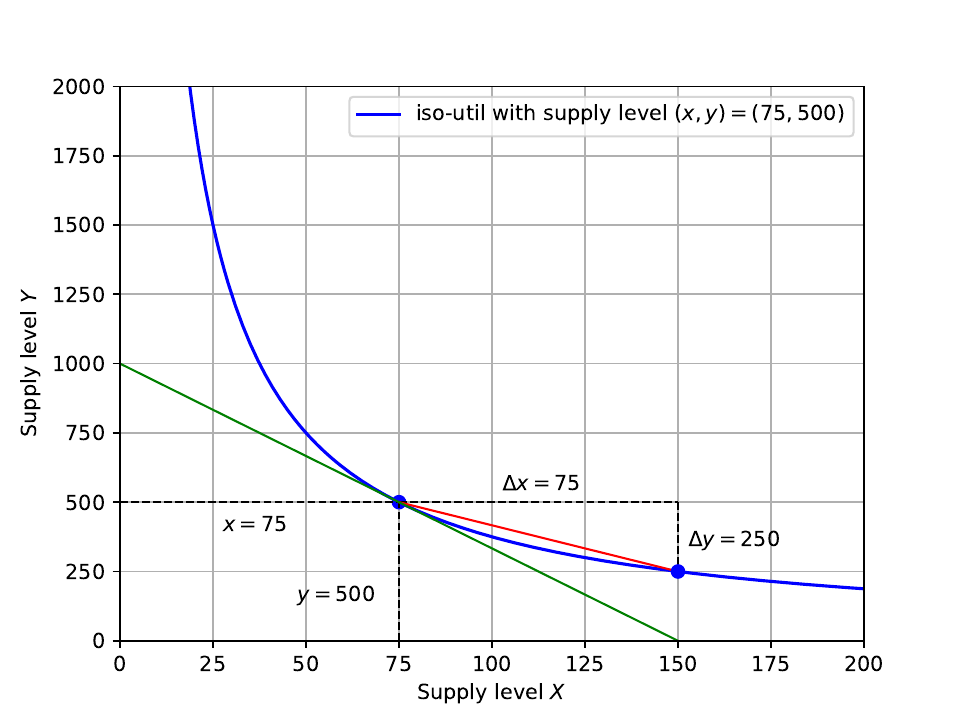}
\caption{Illustration of the realized price $P_{\frac{\myx}{\myy}}(\Delta \myx ) = \frac{\Delta \myx}{ \Delta \myy} = \frac{75}{250} = 0.3$ of one unit of $Y$ in terms of units of $X$ (inverse of the absolute value of the slope of the red secant line) and the corresponding marginal price $P_{\frac{\myx}{\myy}} = - \frac{1}{ f_{+}'(\myx)} = \frac{75}{500} = 0.15$ (inverse of the absolute value of the slope of the green tangent line) of an iso-util supply change in an ideal market with iso-util $x \cdot y = 37500$.}
\label{fig:iso_util_prices}
\end{figure}

\begin{example}
Figure~\ref{fig:LOB_isoutil_settled} below shows an iso-util of a limit order book market. We observe that the iso-util is pieceise linear and has kinks. At the supply levels corresponding to the kinks, the marginal prices of~$X$ and~$Y$ do not coincide. More precisely, at the current supply level~$(x,y)$ the best ask price of the limit order book corresponds to~$P_{\frac{x}{y}}$ which is equal to the inverse of the absolute value of the right derivative of~$f$ at the current supply level. The best bid price of the limit order book corresponds to~$\hat P _{\frac{x}{y}}$ and is equal to the inverse of the absolute value of the left derivative of~$f$. The difference ~$|P_{\frac{\myx}{\myy}} - \hat P_{\frac{\myx}{\myy}}|$ of both prices represents the current bid-ask spread. For more details we refer to Section~\ref{s_isoutil_LOB}.
\end{example}

The next proposition shows how marginal prices can be directly computed via the utility function~$U$.

\begin{proposition}\label{p_gradient_characterization_prices}
We consider a market~$\mathcal{M}$ with utility function~$U$. Let~
\begin{align}
\nabla_{\min} U = (\partial_{x,\min} U, \partial_{y,\min} U)^\top \quad \mbox{and} \quad \nabla_{\max} U = (\partial_{x, \max} U, \partial_{y, \max} U)^\top        
\end{align}
denote the minimal and maximal element of the sub-gradient~$\partial U$. Then the marginal price of the asset~$X$ at supply level~$(x,y)$ is given by  \begin{align}\label{e_marginal_price_of_X_via_utility_function}
\hat P_{\frac{x}{y}} = \frac{\partial_{y,\min} U(x,y)}{\partial_{x,\min} U (x,y)}
\end{align}
and the marginal price of the asset~$Y$ at supply level~$(x,y)$ is given by
\begin{align}\label{e_marginal_price_of_Y_via_utility_function}
P_{\frac{x}{y}} = \frac{\partial_{y,\max} U (x,y)}{\partial_{x,\max} U (x,y)}.
\end{align}
If~$U$ is differentiable at~$(x,y)$ then the marginal prices of~$Y$ and $X$ coincide and are given by 
\begin{align} \label{e_prices_via_utility_function}
P_{\frac{\myx}{\myy}} = \hat P_{\frac{\myx}{\myy}} = \frac{\partial_{\myy} U (x,y)}{\partial_{\myx} U(x,y)}.
\end{align}
\end{proposition}

\begin{remark} [Interpretation of formula~\eqref{e_prices_via_utility_function}] 
We observe that by~\eqref{e_prices_via_utility_function} a small~$\partial_\myx U$ and a large~$\partial_\myy U$ results in a high price. The reason is that more supply of~$\myX$ only yields a small gain in utility, whereas more supply of~$\myY$ yields a large gain of utility. Hence, traders exchange excess of~$\myX$ directly in~$\myY$ resulting in a higher price for~$\myY$. 
\end{remark}
    
\begin{proof}[Proof of Proposition~\ref{p_gradient_characterization_prices}]
Let us focus on deducing the formula in~\eqref{e_marginal_price_of_Y_via_utility_function}. The desired forumla in~\eqref{e_marginal_price_of_X_via_utility_function} follows from a similar argument.  Let the function~$f$ parameterise the iso-util of the market~$\mathcal{M}$ with supply level~$(\myx, \myy)$. Then by Definition~\ref{d_prices} the price~$P_{\frac{\myy}{\myx}}$ is given by
\begin{align}
P_{\frac{\myy}{\myx}}= - f_{+}'(\myx),
\end{align}
where~$f_{+}$ denotes the right derivative of~$f$. Because the function~$f$ is convex it holds that 
\begin{align}
f_{+}' = -\frac{\partial_{x,\max} U}{\partial_{y,\max} U}.
\end{align}
Using the formula~$P_{\frac{\myy}{\myx}} = \left( P_{\frac{\myy}{\myx}}\right)^{-1}$ yields the desired result. Finally, suppose that~$U$ is differentiable. Then the sub-gradient consists of exactly one element, namely the gradient, and the desired formula in~\eqref{e_prices_via_utility_function} readily follows.
\end{proof}    

\begin{corollary}[Monotonicity of prices]\label{p_monotonicity_of_prices}
We consider a market~$\mathcal{M}$ with utility function~$U$. If the utility function~$U$ is concave, then the marginal price $P_{\frac{\myx}{\myy} }$ is increasing in~$\myx$ and decreasing in~$\myy$. If the the utility function~$U$ is strictly concave then the marginal price $P_{\frac{\myx}{\myy} }$ is strictly increasing in~$\myx$ and strictly decreasing in~$\myy$.
\end{corollary}

\begin{proof}[Proof of Corollary~\ref{p_monotonicity_of_prices}]
We observe that because the utility function~$U$ is concave it follows that for fixed~$\myy$  the function $\myx \mapsto \partial_\myx U (\myx, \myy)$ is decreasing as a function of~$\myx$, and similarly that $\myy \mapsto \partial_\myy U (\myx, \myy)$ is a decreasing as a function of~$\myy$. Hence, the formula~\eqref{e_prices_via_utility_function} yields the desired statement.  
\end{proof}

\begin{remark}
The statement of Corollary~\ref{p_monotonicity_of_prices} does not hold if the assumptions are relaxed to quasi concave utility functions~$U$. However, this is not very restrictive as in many situations quasi concavity can be strengthened to concavity (see also Remark~\ref{r_conditions_markets}).
\end{remark}


\section{Limit order book markets}\label{s_isoutil_LOB}

As discussed in Remark~\ref{rem:d_explicit_markets}, limit order book markets are semi-explicit. The utility function is not observable, but the current iso-util can be derived from a snapshot of the limit order book. It turns out that this correspondence is one-to-one. This section is partly inspired by~\cite{Young:20}, where the relation between automated market makers and limit order book markets is studied for smooth iso-utils. As limit order books are naturally discrete objects and give rise to piecewise linear iso-utils, we extend the analysis to the non-smooth case. This is important because the jumps of the slopes in non-smooth iso-utils contain important information, e.g., the current bid-ask spread. In our approach, we define the limit order book via supply and demand measures covering the smooth and non-smooth case simultaneously. For an extensive treatment of market microstructure and limit order books we refer to~\cite{Bouchaud:18}.

\subsection{Supply and demand measures and limit order books}

Supply and demand measures model the supply and demand of a certain asset for a given price range. We will use them to define limit order books. The advantage of this approach is that it can distinguish between settled and unsettled limit order books. In unsettled limit order books there is a price-disagreement in the sense of overlapping buy and sell limit orders. Unsettled limit order books arise often. For example, the opening and closing auctions at a stock exchange result in an unsettled limit order book. When aggregating markets (see Section~\ref{sec:aggregate_markets} below) unsettled limit order books also appear naturally. 

\begin{remark}
In this section, we fix a market~$\mathcal{M}$ of the asset pair~$(X,Y)$. We think of the asset~$X$ as the num\'eraire, e.g., US dollar (see Remark~\ref{rem:numeraire}).
\end{remark}

\begin{definition}[Supply and demand measures~$\mu_s$ and~$\mu_d$ and the limit order book~$\mathscr{L}$.]\label{def:supply_demand_measures} A positive Borel measure~$\mu_s$ on the space~$(0, \infty)$ is called supply measure if
\begin{align}\label{eq:d_supply_measure}
\lim_{p \to 0} \mu_s ((0,p]) =0.
\end{align}
A positive Borel measure~$\mu_d$ on the space~$(0, \infty)$ is called demand measure if
\begin{align}\label{eq:d_demand_measure}
\lim_{p \to \infty} \mu_d ((p, \infty)) =0.
\end{align}
The associated unsettled limit order book~$\mathscr{L}$ is given by the ordered pair~$\mathscr{L}= (\mu_d, \mu_s)$. We say that the limit order book is settled if there is a number~$m \in [0, \infty)$, called mid price, such that $\supp(\mu_d) \subset (0, m]$ and~$\supp(\mu_s) \subset [m, \infty)$. Here $\supp(\mu)$ denotes the support of the measure $\mu$. The best bid price is defined as $p_b:=\sup \supp \mu_d$ and the best ask price as~$p_a:=\inf \supp \mu_s$. In particular, the limit order book is settled if and only if $p_b \leq p_a$.
\end{definition}

The measures $\mu_s$ and $\mu_d$ are also sometimes referred to as volume profile or depth profile in the literature.

\begin{remark}
We point out that in Definition~\ref{def:supply_demand_measures} a settled (or cleared) limit order book includes the special case where there is no bid-ask spread, i.e., best-bid and best-ask price are allowed to coincide. This becomes relevant for the notion of iso-util clearing of a limit order book introduced in Section~\ref{subsec:clearing}, as well as when we associate a smooth iso-util to a limit order book as discussed in Section~\ref{subsec:equivalence}.   
\end{remark}

\begin{remark}[Interpretation of supply and demand measures~$\mu_s$ and~$\mu_d$]
For a Borel-measurable subset~$A \subset (0,\infty)$ the number~$\mu_d(A) \geq 0$ represents the total buy-side volume of asset~$Y$ at prices in the set~$A$. Similarly, the number~$\mu_s(A) \geq 0$ represents the total sell-side volume of asset~$Y$ at prices in the set~$A$. The interpretation of~\eqref{eq:d_supply_measure} is that as the price decreases to zero the supply of the asset~$Y$ vanishes. The interpretation of~\eqref{eq:d_demand_measure} is that as the price increases to infinity the demand for the asset~$Y$ vanishes.
\end{remark}

\begin{example}
We give a hypothetical example of a discrete settled and unsettled limit order book in Table~\ref{tab:discrete_LOBs}. The unsettled values (left) are chosen to be extreme on purpose to get a better graphical illustration. Let us focus on the settled limit order book (right). Each limit order consists of three components: the sign of the order (buy or sell order); quantity or volume of the asset; and the limit price. Limit buy orders appear on the bid side and limit sell orders appear on the ask side. The corresponding demand measure~$\mu_d$ of the limit order book is given by 
\begin{align}
\mu_d = 50 \delta_{10} + 30 \delta_{40} + 20 \delta_{80} + 10 \delta_{94} + 12 \delta_{100},
\end{align}
where~$\delta_x$ denotes the Dirac measure at the point~$x$. The supply measure $\mu_s$ is given by
\begin{align}
\mu_s = 12 \delta_{110} + 20 \delta_{140} + 30 \delta_{170} + 50 \delta_{250} + 50 \delta_{500}. 
\end{align}
The limit order book is settled as the support of the demand and supply measures satisfy
\begin{align}
\supp(\mu_d) \subset (0, 100] \qquad \mbox{and} \qquad \supp(\mu_s) \subset [110, \infty).
\end{align}
Therefore, the mid price~$m$ can be chosen as any number in the bid-ask spread~$m \in (100, 110)$ with best bid price at 100 and best ask price at 110; i.e., the mid price~$m$ is not unique. Often, the convention is to use the midpoint of the bid-ask spread, which in our case would correspond to choosing~$m=105$. 
\end{example}

\begin{table}[ht]
\centering
{\small
\begin{tabular}{cccc}
\toprule
\multicolumn{2}{c}{{\color{blue}Bid}} & \multicolumn{2}{c}{\color{red}Ask} \\
\midrule
Price & Quantity & Price & Quantity \\
\midrule
\$300 & 10 & \$50 & 10 \\
\$135 & 20 &  \$100 & 12 \\
\$110 & 19 &  \$105 & 14\\
\$100 & 12 & \$110 & 25 \\
\$94 & 10 & \$140 & 20 \\
\$80 & 20 & \$170 & 30 \\
\$40 & 30 & \$250 & 50 \\
\$10 & 50 & \$500 & 50 \\
\bottomrule
\end{tabular}
\hspace{1em}
\begin{tabular}{cccc}
\toprule
\multicolumn{2}{c}{{\color{blue}Bid}} & \multicolumn{2}{c}{\color{red}Ask} \\
\midrule
Price & Quantity & Price & Quantity \\
\midrule
\$100 & 12 & \$110 & 12 \\
\$94 & 10 & \$140 & 20 \\
\$80 & 20 & \$170 & 30 \\
\$40 & 30 & \$250 & 50 \\
\$10 & 50 & \$500 & 50 \\ \\ \\ \\
\bottomrule
\end{tabular}
}
\caption{Example of a settled limit order book (right) which results from the adiabatic clearing introduced in Proposition~\ref{prop:settling_a_market} below of the unsettled limit order book (left).}
\label{tab:discrete_LOBs}
\end{table}

\begin{example} \label{ex:LOB_BTCUSDT_measures}
Figure~\ref{fig:LOB_BTCUSDT_measures} illustrates the discrete supply measure $\mu_s$ (red) and the discrete demand measure $\mu_d$ (blue) of Binance's limit order book for the asset pair $(X,Y) = (\text{USDT},\text{BTC})$, which was observed on January 2, 2024 at 05:21:16 a.m. (UTC). Best bid price is $45,265.81$, best ask price is $45,265.82$, and the depicted price range is from $45,000.00$ to $45,500.00$.  
\end{example}

\begin{figure}[ht]
\centering
 \includegraphics[scale=0.65]{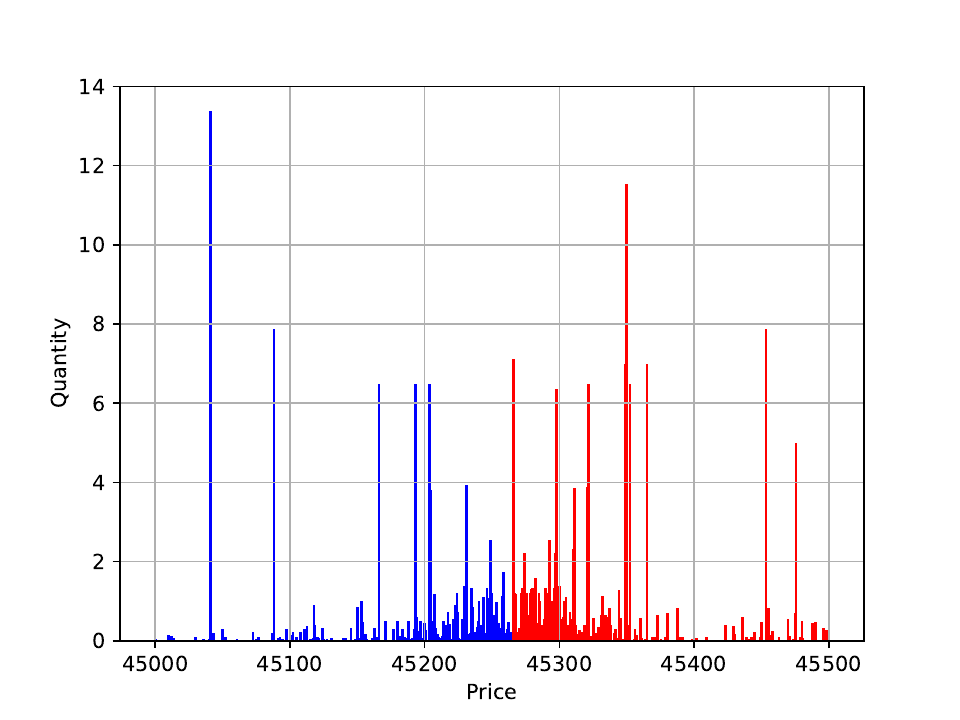} 
\caption{Illustration of the discrete supply measure $\mu_s$ (red) and discrete demand measure $\mu_d$ (blue) for the asset pair (USDT, BTC) traded on Binance's limit order book on January 2, 2024 at 05:21:16 AM (UTC); best bid price is 45,265.81 USDT and best ask price is 45,265.82 USDT.}
\label{fig:LOB_BTCUSDT_measures}
\end{figure}

Conditions~\eqref{eq:d_supply_measure} and~\eqref{eq:d_demand_measure} allow an alternative characterization of supply and demand measures via the \emph{Remaining Demand Function} (RDF) and the \emph{Remaining Supply Function} (RSF). 

\begin{definition}[Remaining demand and remaining supply functions]\label{def:remaining_demand_supple_function}
A function~$F_d: (0, \infty) \to (0, \infty) $ is called remaining demand function (RDF) if it satisfies the following conditions:
\begin{enumerate}
\item[(i)] $F_d$ is non-increasing,
\item[(ii)] $F_d$ is right-continuous,
\item[(iii)] $F_d$ vanishes at infinity, i.e., $\lim_{p \to \infty} F_d(p) =0$.
\end{enumerate}
A function~$F_s: (0, \infty) \to (0, \infty) $ is called remaining supply function (RSF) if it satisfies the following conditions:
\begin{enumerate}
\item[(i)] $F_s$ is non-decreasing,
\item[(ii)] $F_s$ is right-continuous,
\item[(iii)] $F_s$ vanishes at~$0$, i.e., $\lim_{p \to 0} F_s(p) =0$.
\end{enumerate}
We also use the notation $RDF(p) := F_d(p)$ and $RSF(p) := F_s(d)$. 
\end{definition}

In the literature, RSF and RDF are also referred to as bid and ask depth or liquidity curves. They play a similar role as the cumulative distribution function and the complementary cumulative distribution function (or tail distribution function) which are known to characterize Borel probability measures on the real line. 

\begin{proposition}[Characterization of supply measures via remaining supply functions]\label{prop:supply_measures_rsf}
Consider a Borel measure~$\mu_s$ on the space~$(0, \infty)$. Then~$\mu_s$ is a supply measure if and only if the function~$F_s:(0, \infty) \to (0, \infty) $ given by~$F_s(p) := \mu_s ((0,p]) $ is a remaining supply function in the sense of Definition~\ref{def:remaining_demand_supple_function}.
\end{proposition}

\begin{proposition}[Characterization of demand measures via remaining demand functions]\label{prop:demand_measures_rdf}
Consider a Borel measure~$\mu_d$ on the space~$(0, \infty)$. Then~$\mu_d$ is a demand measure if and only if the function~$F_d:(0, \infty) \to (0, \infty) $ given by~$F_d(p) := \mu_d ((p, \infty)) $ is a remaining demand function in the sense of Definition~\ref{def:remaining_demand_supple_function}.
\end{proposition}

\begin{proof}[Proof of Proposition~\ref{prop:supply_measures_rsf} and~\ref{prop:demand_measures_rdf}]
The proof is straightforward. One reproduces the arguments that probability measures on the real line are characterized by cumulative distribution functions. The main difference is that one needs to use condition~\eqref{eq:d_supply_measure} instead of the fact that probability measures have finite overall mass. We leave the details as an exercise.
\end{proof}

\begin{remark}[Interpretation of RDF and RSF]
Note that the remaining demand function $F_d(p)$ represents the total volume of limit buy orders with a price strictly greater than $p$. The remaining supply function $F_s(p)$ represents the total volume of limit sell orders with a price less than or equal to $p$.
\end{remark}

Using Proposition~\ref{prop:supply_measures_rsf} and Proposition~\ref{prop:demand_measures_rdf}, we get the following  immediate characterization of settled limit order books.

\begin{corollary}\label{prop:settled_lob_via_rsf_rdf}
A limit order book~$\mathscr{L}$ is settled if and only if there is a number~$m \in (0, \infty)$ such that~$\lim_{m^+ \downarrow m} F_d(m^+)= \lim_{m^- \uparrow m} F_s(m^-) =0 $, where $F_d$ and $F_s$ denote the LOB's remaining demand and supply function, respectively. In this case the number~$m$ is a mid price of the settled limit order book. 
\end{corollary}

\begin{example}[Unsettled and settled limit order book]
In Figure~\ref{fig:RDF_RSF_smooth}, we give an illustrative example of smooth unsettled remaining demand and remaining supply functions, which overlap, as well as their settled counterparts after the adiabatic clearing introduced in Proposition~\ref{prop:settling_a_market} below. Similarly, in Figure~\ref{fig:RDF_RSF_discrete} we depict the piecewise linear unsettled and settled remaining demand and remaining supply functions, which correspond to the limit order books in Table~\ref{tab:discrete_LOBs}.
\end{example} 

\begin{figure}[ht]
\centering
\includegraphics[scale=0.6]{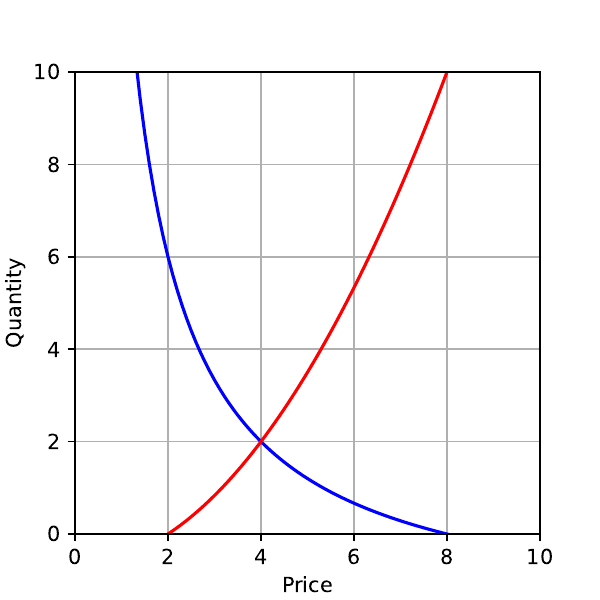}
\includegraphics[scale=0.6]{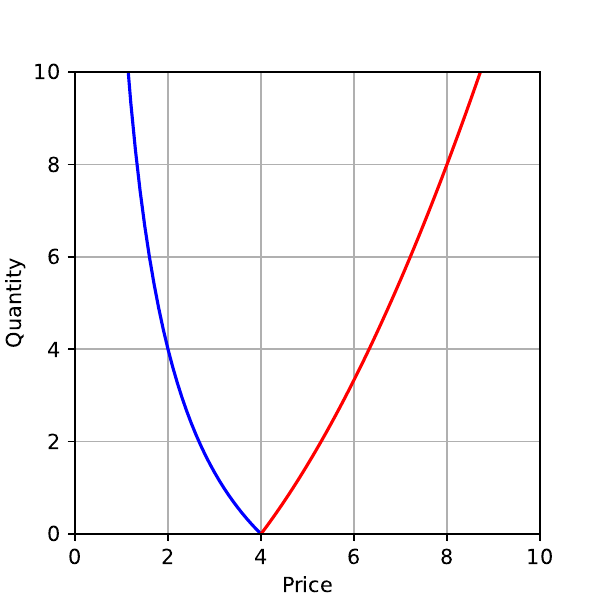}
\caption{Illustration of a smooth remaining demand function (blue) and remaining supply function (red) of a settled limit order book (right panel) which results from the adiabatic clearing from Proposition~\ref{prop:settling_a_market} of the unsettled limit order book (left panel).}
\label{fig:RDF_RSF_smooth}
\end{figure}

\begin{figure}[ht]
\centering
\includegraphics[scale=0.614]{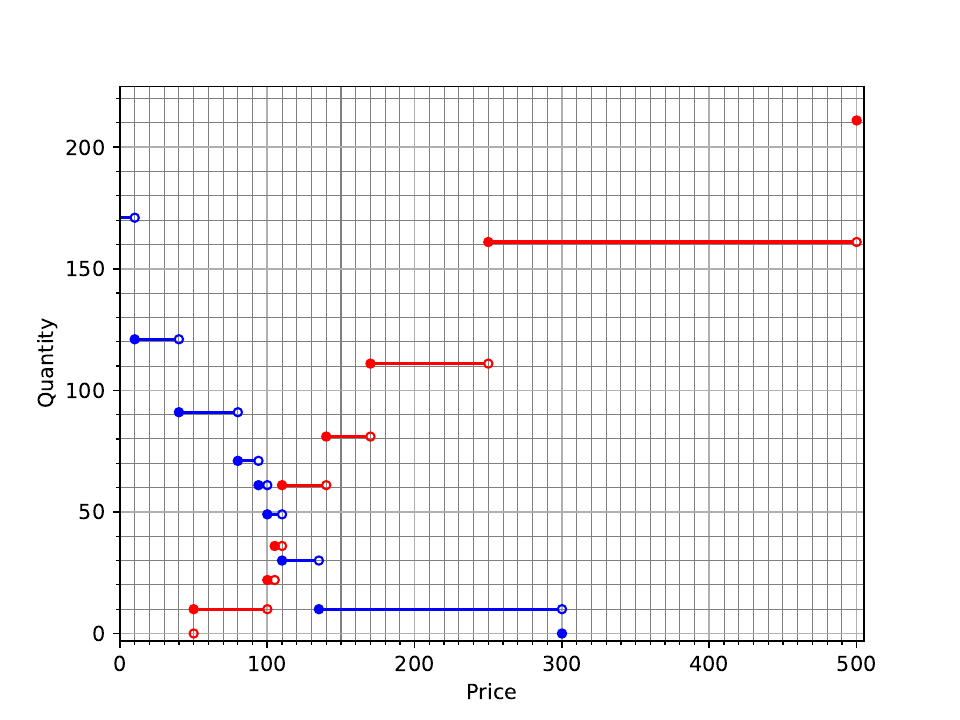}
{\scriptsize
\begin{tabular}{|c|c|c|c|c|c|c|c|c|c|c|c|c|c|c|c|}
\hline
price & 0 & 10 & 40 & 50 & 80 & 94 & 100 & 105 & 110 & 135 & 140 & 170 & 250 & 300 & 500 \\ \hline
$RDF$ & 171 & 121 & 91 & 91 & 71 & 61 & 49 & 49 & 30 & 10 & 10 & 10 & 10 & 0 & 0 \\
$RSF$ & 0 & 0 & 0 & 10 & 10 & 10 & 22 & 36 & 61 & 61 & 81 & 111 & 161 & 161 & 211 \\ \hline
\end{tabular}
}
\\
\includegraphics[scale=0.614]{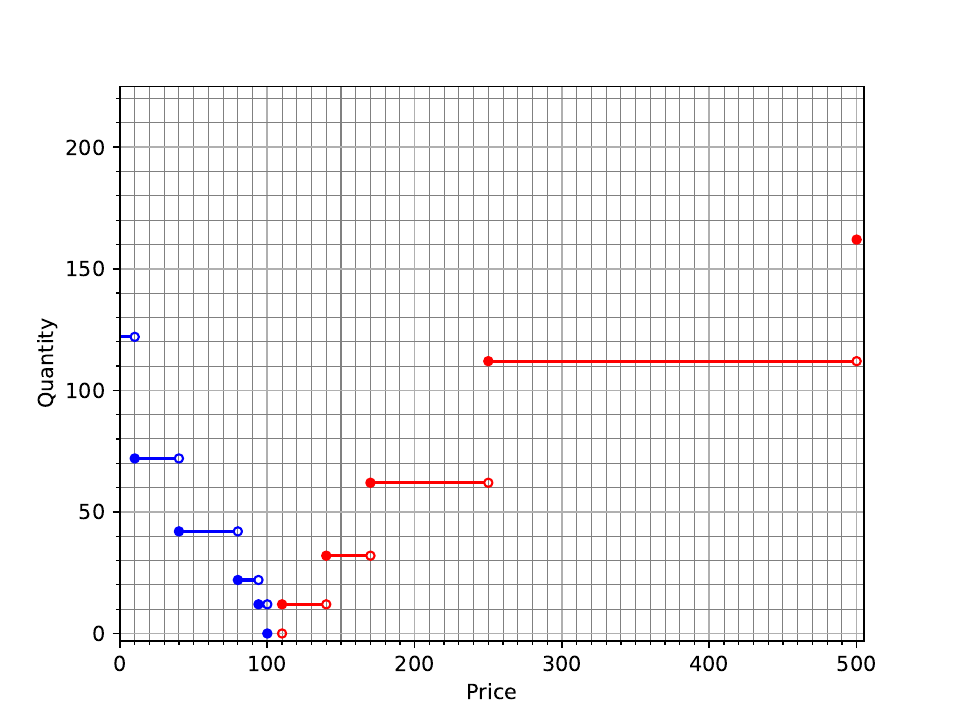}
{\scriptsize
\begin{tabular}{|c|c|c|c|c|c|c|c|c|c|c|c|c|c|c|c|}
\hline
price & 0 & 10 & 40 & 50 & 80 & 94 & 100 & 105 & 110 & 135 & 140 & 170 & 250 & 300 & 500 \\ \hline
$RDF$ & 122 & 72 & 42 & 42 & 22 & 12 & 0 & 0 & 0 & 0 & 0 & 0 & 0 & 0 & 0 \\
$RSF$ & 0 & 0 & 0 & 0 & 0 & 0 & 0 & 0 & 12 & 12 & 32 & 62 & 112 & 112 & 162 \\ \hline
\end{tabular}
}
\caption{Illustration of the piecewise linear remaining demand function (blue) and remaining supply function (red) of the settled limit order book (lower panel) from Table~\ref{tab:discrete_LOBs} (right) which results from the adiabatic clearing from Proposition~\ref{prop:settling_a_market} of the unsettled limit order book (upper panel) from Table~\ref{tab:discrete_LOBs} (left).}
\label{fig:RDF_RSF_discrete}
\end{figure}

\begin{remark}[Treatment of limit order books in~\cite{Young:20}]
Let us compare our approach to the analysis in~\cite{Young:20}. Therein, only settled limit order books are considered, which implies by Corollary~\ref{prop:settled_lob_via_rsf_rdf} that the remaining demand function~$F_d$ and remaining supply function~$F_s$ have disjoint support. As a consequence, the settled limit order book in~\cite{Young:20} is characterized by the function~$G = F_d+F_s$, i.e., a function~$G: (0, \infty) \to (0, \infty)$ that satisfies the following conditions:
\begin{enumerate}
\item[(i)] $G$ is right-continuous;
\item[(ii)] there exists a point~$m \in (0, \infty)$, namely the midpoint, such that
\begin{itemize}
\item $G$ is non-increasing on~$(0,m)$;
\item it holds~$G(m)=0$;
\item $G$ is non-decreasing~$(0,m)$.
\end{itemize}
\end{enumerate}
\end{remark}

\subsection{Adiabatic and iso-util clearing of a limit order book}
\label{subsec:clearing}

Clearing describes the process of transforming an unsettled limit order book $\mathscr{L}$ into a settled limit order book $\mathscr{L}_\sigma$. When the limit order book is unsettled there are two canonical ways for clearing it. We call the first mechanism {\bf adiabatic clearing} and the second mechanism {\bf iso-util clearing}. \\

First, we consider the adiabatic clearing process. In this mechanism, overlapping buy and sell limit orders are matched by an arbitrageur who cashes in the price difference as a risk-less profit, and the orders are removed from the limit order book. To make this more precise, let~$p>0$ and recall that the left-limit~$\lim_{\hat p \uparrow p} RDF (\hat p)$ denotes the number of units of~$Y$ that can be sold to limit buy orders with a price higher than or equal to~$p$ (note that $RDF$ is defined to be right-continuous). Similarly,~$RSF(p)$ denotes the number of units that can be bought from limit sell orders with a price lower than or equal to~$p$. We define the quantity
\begin{align} \label{def:auction_volume}
Z(p) := \min \left\{ \lim_{\hat p \uparrow p} RDF(\hat p), RSF(p)   \right\}.
\end{align}  
In a settled limit order book we have $Z(p)=0$ for all~$p>0$. In contrast, if the market is unsettled, some traders misprice the asset~$Y$ in the sense that some are willing to buy at higher prices than those at which others are willing to sell, and we have~$Z(p) >0$ for some~$p>0$. This means that $Z(p)$ denotes the total volume that can be offset by an arbitrageur for a risk-free profit. In this case, the arbitrageur will seek to maximize her profit. This corresponds to finding the clearing price
\begin{align} \label{def:auction_price}
p_e := \arg \max_p Z(p),
\end{align}
and then to trade all limit sell orders with a price lower than $p_e$ against all limit buy orders with a price higher than $p_e$. This is straightforward if the maximizer~$p_e$ is unique. The next lemma studies the case where the maximizer~$p_e$ is not unique. 

\begin{lemma}\label{prop:clearing_auxiliary_lemma}
We consider an unsettled limit order book~$\mathscr{L}$ with an unsettled remaining demand function~$RDF_u$ and an unsettled remaining supply function~$RSF_u$. Let  
\begin{align}
p_d := \inf \left\{ p >0 \ | \ RDF_u(p) \leq RSF_u(p) \right\} 
\end{align}
denote the smallest possible price such that supply overcomes demand and
\begin{align}
p_s := \sup \left\{ p >0 \ | \ RDF_u(p) \geq RSF_u(p) \right\} 
\end{align}
the highest possible price such that demand overcomes supply. \newline
    
Then $p_d \leq p_s$, and $p_d < p_s$ implies that 
\begin{align}\label{equ:p_d_neq_p_s}
RDF_u(p)= RSF_u(p) \quad \text{for all} \quad p_d < p < p_s.
\end{align} 
It also holds that
\begin{align}\label{equ:clearing_volume}
Z(p_s) = Z(p_d) = Z(p) \quad \text{for all } \quad p_d < p < p_s 
\end{align}
and
\begin{align}\label{equ:clearing_volume_2}
Z(p_s) = \max_p Z(p). 
\end{align}
We call $Z:= Z(p_s)$ the clearing volume of the asset~$Y$.
\end{lemma}
\begin{proof}[Proof of Lemma~\ref{prop:clearing_auxiliary_lemma}]
The statements~$p_d \leq p_s$ and~\eqref{equ:p_d_neq_p_s} follow directly from the observation that the function~$RDF$ is non-increasing, the function~$RSF$ is non-decreasing, and the definitions of~$p_d$ and~$p_s$ are equivalent to
\begin{align}
p_d = \sup \left\{ p >0 \ | \ RDF_u(p) > RSF_u(p)    \right\}
\end{align}
and
\begin{align}
p_s = \inf \left\{ p >0 \ | \ RDF_u(p) <  RSF_u(p)    \right\}. 
\end{align}
Let us now turn to the identity in~\eqref{equ:clearing_volume}. If~$p_d = p_s$ there is nothing to show. Hence, we assume~$p_d < p_s$. The desired statements in~\eqref{equ:clearing_volume} and~\eqref{equ:clearing_volume_2} then follow from a straightforward argument using right-continuity and monotonicity of $RDF_u$ and $RSF_u$, as well as the identity in~\eqref{equ:p_d_neq_p_s}.
\end{proof}

\begin{proposition} [Adiabatic clearing of a limit order book]\label{prop:settling_a_market}
We consider the situation of Lemma~\ref{prop:clearing_auxiliary_lemma} with an unsettled limit order book $\mathscr{L}$. Then
\begin{align}\label{e_adiabatic_RDF}
RDF_a(p) := 
\begin{cases}
\max \{ RDF_u(p) - Z ,0 \} & \text{for } p < p_d,\\
0 & \text{for } p \geq p_d,
\end{cases}
\end{align}
and
\begin{align}\label{e_adiabatic_RSF}
RSF_a(p) := 
\begin{cases}
0  & \text{for } p < p_s, \\
\max \{ RSF_u(p) - Z , 0 \} & \text{for } p \geq p_s, 
\end{cases}
\end{align}
define a remaining demand function and remaining supply function of the settled limit order book $\mathscr{L}_\sigma$. Moreover,~$p_d$ and~$p_s$ denote the best bid and best ask price, respectively.
\end{proposition}

The best way to understand the formulas~\eqref{e_adiabatic_RDF} and~\eqref{e_adiabatic_RSF} is to observe that the quantity~$Z$ denotes the total volume of~$Y$ which was matched and offset by the arbitrageur's trading. As a consequence, this amount has to be subtracted from the limit order book. The proof of Proposition~\ref{prop:settling_a_market} is straightforward and omitted.

\begin{definition}
The procedure described in Proposition~\ref{prop:settling_a_market} is called adiabatic settlement or clearing, and we say that~$\mathscr{L}_\sigma$ is the adiabatic-settled limit order book of the unsettled limit order book~$\mathscr{L}$.
\end{definition}

\begin{example}\label{ex:clearing}
The settled limit order book in Table~\ref{tab:discrete_LOBs} (right) is the adiabatic clearing of the unsettled limit order book in Table~\ref{tab:discrete_LOBs} (left). The corresponding unsettled remaining demand and supply functions $RDF_u$ and $RSF_u$ as well as their adiabatic-settled counterparts $RDF_a$ and $RSF_a$ are depicted in Figure~\ref{fig:RDF_RSF_discrete}. The adiabatic clearing procedure is also described in Table~\ref{tab:example_adiabatic_isoutil} (left panel). Therein, the clearing volume is $Z=49$ and all crossed-out limit buy and sell orders are offset by the arbitrageur.
\end{example}

\begin{table}[ht]
\centering
{\small
\begin{tabular}{cccc}
\toprule
\multicolumn{2}{c}{{\color{blue}Bid}} & \multicolumn{2}{c}{\color{red}Ask} \\
\midrule
Price & Quantity & Price & Quantity \\
\midrule
\$300 & \textbf{\color{black}\cancel{10}} 0 & \$50 & \textbf{\color{black}\cancel{10}} 0 \\
\$135 & \textbf{\color{black}\cancel{20}} 0 &  \$100 & \textbf{\color{black}\cancel{12}} 0 \\
\$110 & \textbf{\color{black}\cancel{19}} 0 &  \$105 & \textbf{\color{black}\cancel{14}} 0 \\
\$100 & 12 & \$110 & \textbf{\color{black}\cancel{25}} 12 \\
\$94 & 10 & \$140 & 20 \\
\$80 & 20 & \$170 & 30 \\
\$40 & 30 & \$250 & 50 \\
\$10 & 50 & \$500 & 50 \\
\bottomrule
\end{tabular}
\hspace{1em}
\begin{tabular}{cccc}
\toprule
\multicolumn{2}{c}{{\color{blue}Bid}} & \multicolumn{2}{c}{\color{red}Ask} \\
\midrule
Price & Quantity & Price & Quantity \\
\midrule
\$110 & \textbf{13} & \$110 & 12 + \textbf{19} = 31 \\
\$105 & \textbf{14} & \$135 & \textbf{20} \\
\$100 & 12+\textbf{12} = 24 & \$140 & 20 \\
\$94 & 10 & \$170 & 30 \\
\$80 & 20 &  \$250 & 50 \\
\$50 & \textbf{10} & \$300 & \textbf{10} \\
\$40 & 30 & \$500 & 50 \\
\$10 & 50 & & \\
\bottomrule 
\end{tabular}
}
\caption{\emph{Left:} Adiabatic clearing procedure of the unsettled limit order book in Table~\ref{tab:discrete_LOBs} (left). Bold crossed-out orders are matched by the arbitrageur. \emph{Right:} Resulting limit order book from the iso-util clearing of the unsettled limit order book in Table~\ref{tab:discrete_LOBs} (left). Bold quantity values indicate the volume of limit orders that were flipped from one side of the order book to the other (cf.~also the bold crossed-out limit orders in the left table).} 
\label{tab:example_adiabatic_isoutil}
\end{table}

\begin{remark}[Comparison of adiabatic clearing to opening and closing auctions]
The adiabatic clearing mechanism is very similar to the procedure of an opening and closing auction on a stock exchange where an algorithm computes the auction price~$p_e$ in~\eqref{def:auction_price} in order to maximize the total clearing volume when overlapping buy and sell orders are matched. In contrast to the adiabatic settlement, all overlapping buy and sell orders are then executed at the same auction price $p_e$.   
\end{remark}

Next, we describe the iso-util clearing. This mechanism adds an additional step to the adiabatic clearing where the matched limit orders in fact \emph{re-appear} on the opposite side of the limit order book. 

\begin{proposition} [Iso-util clearing of a limit order book]\label{prop:iso_util_settling_a_market}
We consider an unsettled limit order book~$\mathscr{L}$ with an unsettled remaining demand function~$RDF_u$ and an unsettled remaining supply function~$RSF_u$ and its adiabatic-settled counterparts~$RDF_a$ and~~$RSF_a$.  Then 
\begin{align}
\text{RDF}_i(p) = 
\begin{cases}
\text{RDF}_a(p) +Z - \min \left\{ \text{RSF}_u(p) ,Z \right\}  & \text{if } p < p_d, \\
0  & \text{if } p \geq  p_d 
\end{cases} 
\end{align}
and
\begin{align}
\text{RSF}_i(p) = 
\begin{cases}
0  & \text{if } p <  p_s,\\ 
\text{RSF}_a(p) + Z -  \min \left\{ \text{RDF}_u(p) , Z \right\} & \text{if } p \geq p_s
\end{cases}
\end{align}
define a remaining demand function and remaining supply function of the settled limit order book~$\mathscr{L}_\sigma$. Moreover,~$p_d$ and~$p_s$ denote the best bid and best ask price, respectively.
\end{proposition}

The proof is straightforward and omitted.

\begin{definition}
The procedure described in Proposition~\ref{prop:iso_util_settling_a_market} is called iso-util settlement or clearing, and we say that~$\mathscr{L}_\sigma$ is the iso-util-settled limit order book of the unsettled limit order book~$\mathscr{L}$.
\end{definition}

\begin{example}
The limit order book in Table~\ref{tab:example_adiabatic_isoutil} (right) is the iso-util clearing of the unsettled limit order book in Table~\ref{tab:discrete_LOBs} (left). Observe that all crossed-out limit buy and sell orders from the adiabatic clearing in Table~\ref{tab:example_adiabatic_isoutil} (left) just reappear on the opposite side of the book. The resulting remaining demand and supply functions $RDF_i$ and $RSF_i$ are plotted in Figure~\ref{fig:RDF_RSF_discrete_isoutil_clearing} (compare also with Figure~\ref{fig:RDF_RSF_discrete}).
\end{example}

\begin{figure}[ht]
\centering
\includegraphics[scale=0.6]{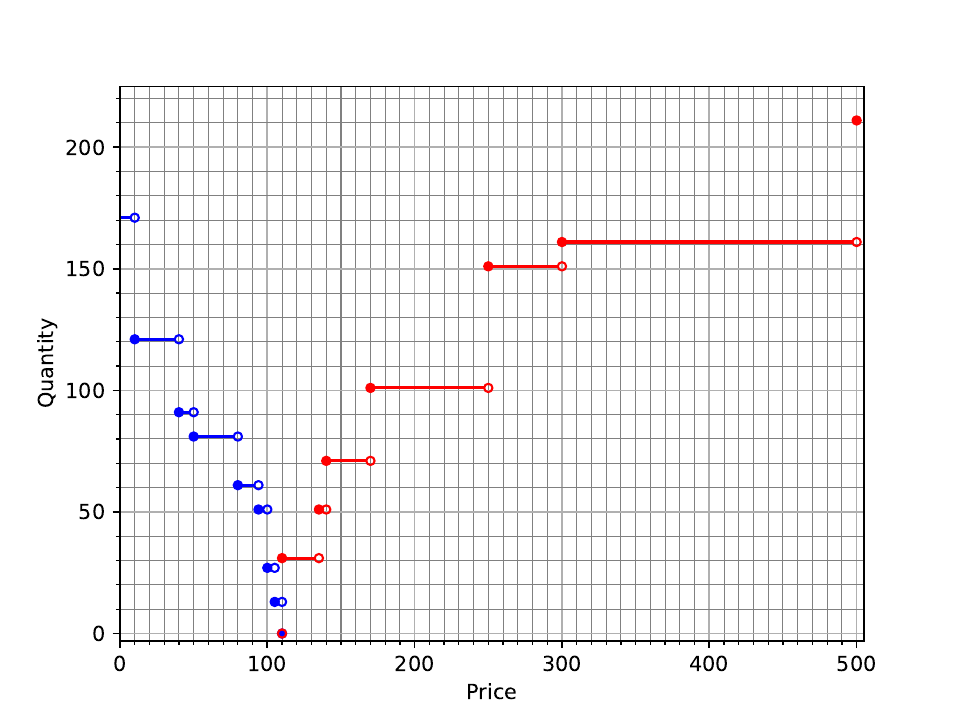}
{\scriptsize
\begin{tabular}{|c|c|c|c|c|c|c|c|c|c|c|c|c|c|c|c|}
\hline
price & 0 & 10 & 40 & 50 & 80 & 94 & 100 & 105 & 110 & 135 & 140 & 170 & 250 & 300 & 500 \\ \hline
$RDF_i$ & 171 & 121 & 91 & 81 & 61 & 51 & 27 & 13 & 0 & 0 & 0 & 0 & 0 & 0 & 0 \\
$RSF_i$ & 0 & 0 & 0 & 0 & 0 & 0 & 0 & 0 & 31 & 51 & 71 & 101 & 151 & 161 & 211 \\ \hline
\end{tabular}
}
\caption{Illustration of the piecewise linear remaining demand function (blue) and remaining supply function (red) of the limit order book in Table~\ref{tab:example_adiabatic_isoutil} (right) after the iso-util clearing of the unsettled limit order book in Table~\ref{tab:discrete_LOBs} (left).}
\label{fig:RDF_RSF_discrete_isoutil_clearing}
\end{figure}

Let us now elaborate on the background and motivation for these two clearing procedures. To this end, we introduce the following notion.

\begin{definition}[Transparent trader] \label{d_transparent_tradent}
We call a market participant a \emph{transparent trader} if the trader communicates her iso-util by posting the associated limit orders in the limit order book. 
\end{definition}

\begin{definition}[Iso-util trader]
We call a market participant an \emph{iso-util trader} if the trader is willing to do any trade which does not decrease her utility. 
\end{definition}

First, we consider what happens if a (buy or sell) limit order of a transparent iso-util trader gets filled. We remind the reader that in this article only the stationary case is studied. That is, we assume that the iso-util of the trader is unaffected by the fact that her limit order gets executed. Since the trader acts iso-util, she would be willing to undo this trade simply because this would leave her utility invariant. Given that the trader is transparent, she would also communicate this intention by submitting a limit order of the same price and same volume, but of opposite type.

\begin{remark} \label{rem:isoutil_liquidprov}
A practical example of iso-util traders is given by liquidity providers in an automated market maker liquidity pool like Uniswap v2 or v3 (see~\cite{Adams:20, Adams:21}). They effectively act as transparent iso-util traders, acting on the iso-util specified in the protocol (with the exception of collected fees, which represent the marginal gain of utility for the liquidity providers). In Uniswap v2, the iso-util is prescribed in the protocol, in Uniswap v3 the user has some freedom to determine her iso-util. For more details we refer to~\cite{EchenimGobetMaurice:24}.
\end{remark}

Suppose a market is comprised of transparent iso-util traders. If the associated limit order book is not settled, the asset is mispriced. As in the adiabatic case, an arbitrageur will take advantage and match overlapping orders for a riskless profit. The main difference to the adiabatic settlement is that the trades of the arbitrageur are executed against limit orders of transparent iso-util traders. Hence, as explained above, these filled limit orders will reappear on the opposite side of the limit order book, explaining the formula of the~$RDF_i$ and~$RSF_i$ in Proposition~\ref{prop:iso_util_settling_a_market}. \\

Let us now shed more light on the adiabatic settlement process and introduce the following notion.

\begin{definition}[Adiabatic trader]
We call a trader $\varepsilon$-adiabatic if the trader is only willing to do trades that increase her utility by at least~$\varepsilon>0$. We call a trader adiabatic if the trader is not willing to undo a trade.
\end{definition}

We describe how a transparent~$\varepsilon$-adiabatic trader acts. Such a trader will only communicate limit orders which increase her utility. Even if small, the trader increases with each trade her utility. In addition, undoing a trade would correspond to a decrease in utility which an adiabatic trader would not tolerate. Therefore, if a limit order of an~$\varepsilon$-adiabatic trader gets filled, it will disappear from the limit order book. Taking the limit~$\varepsilon \to 0$, motivates the definition of an adiabatic trader. The adiabatic clearing process given in Proposition~\ref{prop:settling_a_market} applies to an unsettled market of adiabatic traders. As every trader strives to increase her utility, the adiabatic clearing process is closer to reality; up to the fact that traders might not be transparent and communicate their intentions via limit orders. \\

This raises the question whether iso-utils are a meaningful tool to study financial markets, given the fact that traders typically tend to act adiabatic. Actually, iso-utils are very useful when considering the case of vanishing utility gain. Indeed, the trading curve of an~$\varepsilon$-adiabatic trader approximates her iso-util 
as~$\varepsilon \to 0$. We claim that even in the case of non-small~$\varepsilon$, iso-utils are a useful tool to study the behavior of adiabatic traders, as it may be possible to interpret the trading curve of an adiabatic trader as an iso-util of a new utility function, only that trades are not reversible on this iso-util. \\

The next statement determines the profit from arbitrage when unsettled markets clear. Observe that the arbitrageur behaves similarly in both the adiabatic and iso-util setting. Therefore, the arbitrageur's profit is the same in both clearing mechanisms.

\begin{proposition}[Arbitrage profit in an unsettled market]\label{p_entropy_unsettled_market}
Let us consider an unsettled limit order book $\mathscr{L}$ given by the unsettled remaining supply and demand functions $\text{RSF}_u(p)$ and $\text{RDF}_u(p)$. We define the functions
\begin{align}
A_d(p):= \min\{\text{RDF}_u(p), Z\}  \quad \text{and} \quad   A_s(p) := \min\{\text{RSF}_u(p), Z\}.
\end{align}
When clearing the market, both in an adiabatic or iso-util way, the profit from arbitrage is given by
\begin{align}\label{e_arbitrage_profit}
P = \int_{[p_d,\infty)} p \ d A_d(p)  - \int_{[0,p_s]}  p \ dA_s(p). 
\end{align}
Here, $d A_s(p)$ and $d A_d(p)$ denote the Lebesgue-Stieltjes integral with respect to $A_s(p)$ and $ A_d(p)$. 
\end{proposition}

\begin{proof}[Proof of Proposition~\ref{p_entropy_unsettled_market}]
Arbitrage arises from matching mispriced limit orders. This is independent from the fact that those limit orders vanish out of the limit order book or re-appear on the other side of the book. Hence, the profit from arbitrage is independent of adiabatic or iso-util settlement. The second integral in equation~\eqref{e_arbitrage_profit} denotes the amount of money needed to buy the supply from the under-priced limit sell orders. The first integral denotes the amount of money made from selling this supply to the over-priced limit buy orders. 
\end{proof}

\begin{example}
The arbitrage profit from clearing the unsettled limit order book in Table~\ref{tab:discrete_LOBs} (left) is given by $(10 \cdot 300 + 20 \cdot 135 + 19 \cdot 110) - (10 \cdot 50 + 12 \cdot 100 + 14 \cdot 105 + 13 \cdot 110) = 7790 - 4600 = 3190$; cf. also~Table~\ref{tab:example_adiabatic_isoutil} (left). 
\end{example}

We conclude this section by introducing the terms \textbf{adiabatic} and \textbf{iso-util entropy}.

\begin{definition}[Adiabatic and iso-util entropy] \label{d_adiabatic_iso_util_entropy}
We consider an unsettled limit order book $\mathscr{L}$. The iso-util entropy~$S_i$ is given by the vector
\begin{align} \label{d_adiabatic_iso_util_entropy_s_i}
S_i := (P, 0),
\end{align}
where~$P$ denotes the arbitrage profit given in~\eqref{e_arbitrage_profit}.
The adiabatic entropy~$S_a$ is given by the vector 
\begin{align} \label{d_adiabatic_iso_util_entropy_s_a}
S_a := \left( \int_{[p_d,\infty) } p \ d A_d(p), Z\right),
\end{align}
where $Z$ denotes the clearing volume from Lemma~\ref{prop:clearing_auxiliary_lemma}.
\end{definition}

The adiabatic and iso-util entropy~$S_a$ and~$S_i$ denote the \emph{amount of liquidity that gets lost} by adiabatic or iso-util clearing of an unsettled limit order book. The use of the term \emph{entropy} is inspired by an analogy to the second law of thermodynamics. The entropy plays an important role in the aggregation of financial markets, which is explained in detail in Section~\ref{sec:aggregate_markets}. Its role for economic activity is illustrated in Appendix~~\ref{sec:entropy_in_economics}.\medskip

We observe that the iso-util entropy is smaller than the adiabatic entropy, i.e.~$S_i \leq S_a$. The reason is that in adiabatic settlement, all matched limit orders vanish out of the limit order book. The number~$Z$ denotes how many units of~$Y$ were sold in the clearing process, and the number~$\int_{[p_d,+\infty)} p \ d A_d(p)$ denotes how many dollars were payed for it.

\begin{example}
Simple calculations show that for the unsettled limit order book market in Table~\ref{tab:discrete_LOBs} (left) the adiabatic and iso-util entropies are $S_a = (7790,49)$ and $S_i = (3190,0)$, respectively.
\end{example}

\begin{proposition}\label{p_adiabatic_iso_util_ebntropy}
Consider an unsettled limit order book market $\mathscr{L}$ with unsettled remaining supply and demand functions~$\text{RSF}_u$ and~$\text{RDF}_u$. Let~$(x_u,y_u)$ denote the supply levels of the assets $(X,Y)$ before clearing, i.e.,
\begin{align} \label{p_adiabatic_iso_util_ebntrop_eq}
x_u := \int_{[0,\infty) } p \ d \text{RDF}_u(p) \quad \text{and} \quad y_u := \lim_{p \to \infty} \text{RSF}_u(p).
\end{align}
After clearing the supply levels $(x_a, y_a)$ of the settled market are given by 
\begin{align} \label{p_adiabatic_iso_util_ebntropy_a}
(x_a , y_a) = (x_u,y_u) - S_a
\end{align}
in the case of adiabatic clearing and by
\begin{align} \label{p_adiabatic_iso_util_ebntropy_i}
(x_i, y_i) = (x_u,y_u) - S_i
\end{align}
in the case of iso-util clearing.
\end{proposition}

The proof consists of a straightforward calculation and is omitted. 

\begin{example}
The supply levels of the asset pair $(X,Y)$ in the unsettled limit order book in Table~\ref{tab:discrete_LOBs} (left) are given by $x_u = 300 \cdot 10 + 135 \cdot 20 + 110 \cdot 19 + 100 \cdot 12 + 94 \cdot 10 + 80 \cdot 20 + 40 \cdot 30 + 10 \cdot 50 = 13230$ and $y_u = 10+12+14+25+20+30+50+50 = 211$. Hence, it follows that $(x_a,y_a)=(13230,211) - (7790,49) = (5440,162)$ for the adiabatic clearing and $(x_i,y_i)=(13230,211) - (3190,0) = (10040,211)$ for the iso-util clearing; compare also with the corresponding limit order books in Table~\ref{tab:discrete_LOBs} (right) and Table~\ref{tab:example_adiabatic_isoutil} (right). 
\end{example}

\subsection{Equivalence of iso-utils and limit order books}
\label{subsec:equivalence}

In this section, we will describe how to associate an iso-util to a limit order book and vice versa. For the sake of clarity, we start by explaining the procedure for a simple example first, namely for piecewise linear iso-utils, which turn out to correspond to discrete limit order books. Later, we provide general formulas extending the ones in~\cite{Young:20}.

\begin{figure}[ht]
\centering
\includegraphics[scale=0.65]{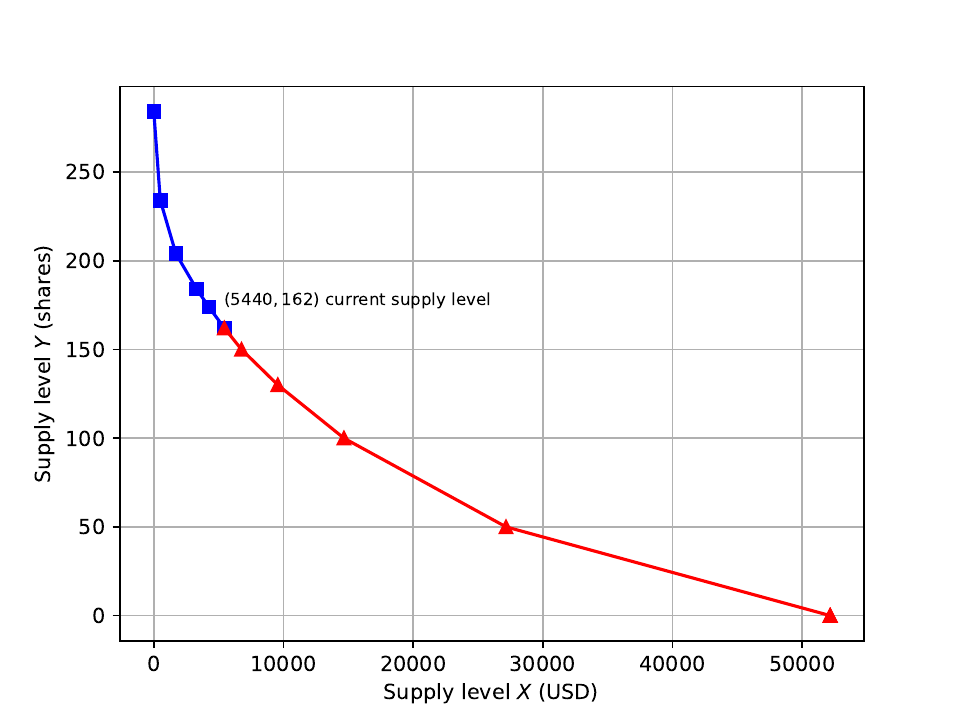} 
\caption{Iso-util generated from the settled limit order book given in Table~\ref{tab:discrete_LOBs} (right).}
\label{fig:LOB_isoutil_settled}
\end{figure}

\begin{table}[ht]
\centering
\caption{Supply levels of the limit order book in Table~\ref{tab:discrete_LOBs} (right).}
\label{tab:supply_levels}
{\small
\begin{tabular}{|c|c|c|c|c|c|c|c|c|c|c|c|}
\hline
 points & 1 & 2 & 3 & 4 & 5 & 6 & 7 & 8 & 9 & 10 & 11 \\
\hline
$x$ & 0 & 500 & 1700 & 3300 & 4240 & 5440 & 6760 & 9560 & 14660 & 27160 & 52160 \\
\hline
$y$ & 284 & 234 & 204 & 184 & 174 & 162 & 150 & 130 & 100 & 50 & 0 \\
\hline
\end{tabular}
}
\end{table}

\begin{example}[Associating a limit order book to a piecewise linear iso-util]
Consider the piecewise linear iso-util illustrated in Figure~\ref{fig:LOB_isoutil_settled}, which results from the linear interpolation of the supply levels summarized in Table~\ref{tab:supply_levels}. The current supply level of the market is given by the 6th data point~$(5440, 162)$. The bid part of the iso-util is marked in blue and the ask part is marked in red (cf.~Definition~\ref{def:bid_ask_part_isoutil}). Observe that with an iso-util trade one could move the supply level from the 6th data point to the 7th, implying that the supply level of $X$ would increase by $6760 -5440= 1320$ dollars and the supply level of $Y$ would decrease by $162-150=12$ units of $Y$. In other words, one could buy $12$ units of $Y$ for $1320$ dollars. Consequently, the first limit sell order in the associated limit order book is $12$ units of $Y$ for $1320/12=110$ dollars each. Similarly, moving from the 7th data point to the 8th would correspond to a limit sell order of $150-130=20$ units of $Y$ at the price of $\frac{9560-6760}{20}=140$ dollars for each unit. Continuing with the remaining data points on the ask part of the iso-util would recover the ask side of the limit order book in Table \ref{tab:discrete_LOBs} (right). Recovering the bid side of the limit order book from the bid part of the iso-util is very similar: An iso-util trade moving the supply levels from the 6th data point to the 5th would imply that $174-162=12$ units of $Y$ could be sold for $5440-4240=1200$ dollars in total, corresponding to a limit buy order of $12$ units of $Y$ for the price of $100$ dollars per unit. Continuing in this way, one would recover the bid side of the limit order book given in Table~\ref{tab:discrete_LOBs} (right).
\end{example}

\begin{example}[Associating an iso-util and supply levels to a discrete limit order book]
Consider the limit order book given in Table~\ref{tab:discrete_LOBs} (right). First, we need to determine the current supply level $(x,y)$. The supply $y$ is given by the total number of shares that can be bought in the limit order book, which corresponds to $y=12+20+30+50+50 = 162$ shares. The supply $x$ is given by the total amount of dollars that can be received by selling as much of the asset $Y$ as possible. In our example, this amounts to $y = 100 \cdot 12 + 94 \cdot 10 + 80 \cdot 20 + 40 \cdot 30 + 10 \cdot 50 = 5440$ dollars. Thus, the current supply level is given by $(5440,162)$ and corresponds to the current liquidity of the asset pair $(X,Y)$ in the limit order book. Next, we determine the pivotal supply levels collected in Table~\ref{tab:supply_levels}. The corresponding iso-util~$I$ is then given by linear interpolation of these points. For example, with the first limit sell order in the limit order book, 12 units of $Y$ can be purchased for $110 \cdot 12 = 1320$ dollars in total. Hence, the point $x=5440+1320 = 6760$ and $y=162-12=150$ is also on the iso-util, recovering the 7th data point. Similarly, one can recover all data points in Table~\ref{tab:supply_levels}, and therefore obtain the iso-util depicted in Figure~\ref{fig:LOB_isoutil_settled}. 
\end{example}

\begin{remark}
We summarize how a discrete limit order book is associated to a piecewise linear iso-util as follows:
\begin{enumerate}
\item[(i)] The current supply level $y$ corresponds to the sum of the volumes of all limit sell orders, and the current supply level $x$ corresponds to the sum of the dollar volumes of all limit buy orders.
\item[(ii)] Each line segment of the iso-util corresponds to a limit order: Line segments in the ask part (i.e., to the right of the current supply level) correspond to limit sell orders (ask side of the book); line segments in the bid part (i.e., to the left of the current supply level) correspond to limit buy orders (bid side of the book).
\item[(iii)] The slope of each line segment corresponds to~$-P_{\frac{y}{x}}$, i.e., the negative of the inverse of the price of $Y$ of the limit order. 
\item[(iv)] The height loss of a line segment corresponds to the volume of the limit order.
\end{enumerate}
\end{remark}

\begin{remark}
In Figure~\ref{fig:LOB_isoutil_settled}, notice that the slope of the line segment on the left-hand side of the current supply level $(x,y) = (5440,162)$ is not the same as the slope of the line segment on the right-hand side. In fact, the difference of the slopes is in one-to-one correspondence with the bid-ask spread in the associated limit order book.
\end{remark}

\begin{example}\label{ex:binance_BTC_USDT}
Figure~\ref{fig:LOB_BTCUSDT_isoutil} illustrates the iso-util of Binance's limit order book of the asset pair $(X,Y)=(\text{USDT}, \text{BTC})$ observed on January 2, 2024 at 05:21:16 a.m.~(UTC); cf.~also Example~\ref{ex:LOB_BTCUSDT_measures}. Best bid is at $45,265.81$ and best ask is at $45,265.82$. The current supply level of the entire order book is $x=170,595,915.88$ USDT and $y=4,078.43$ BTC. The upper panel illustrates the piecwise linear iso-util for the first 100 levels on both sides of the book with markers (squares and triangles) at the different price levels. Observe that the first line segment of the supply part (red) is much longer than the first line segment of the demand part (blue). This is due to an order imbalance at the best prices. Indeed, the available volume at the best ask price is $7.12350$ BTC whereas the available volume at the best bid price is only $0.72667$ BTC. The lower panel shows the iso-util of the entire limit order book together with a tangent line at the current supply level $(x,y)$ with slope $-\frac{1}{45265.815}$, i.e., with midprice $45,265.815$ as marginal price. Because of the small tick size between the different available price levels, the slopes between the line segments of the iso-util are almost identical. However, compared to the linear tangent line, we note that the iso-util is non-linear.
\end{example}

\begin{figure}[ht]
\centering
\includegraphics[scale=0.65]{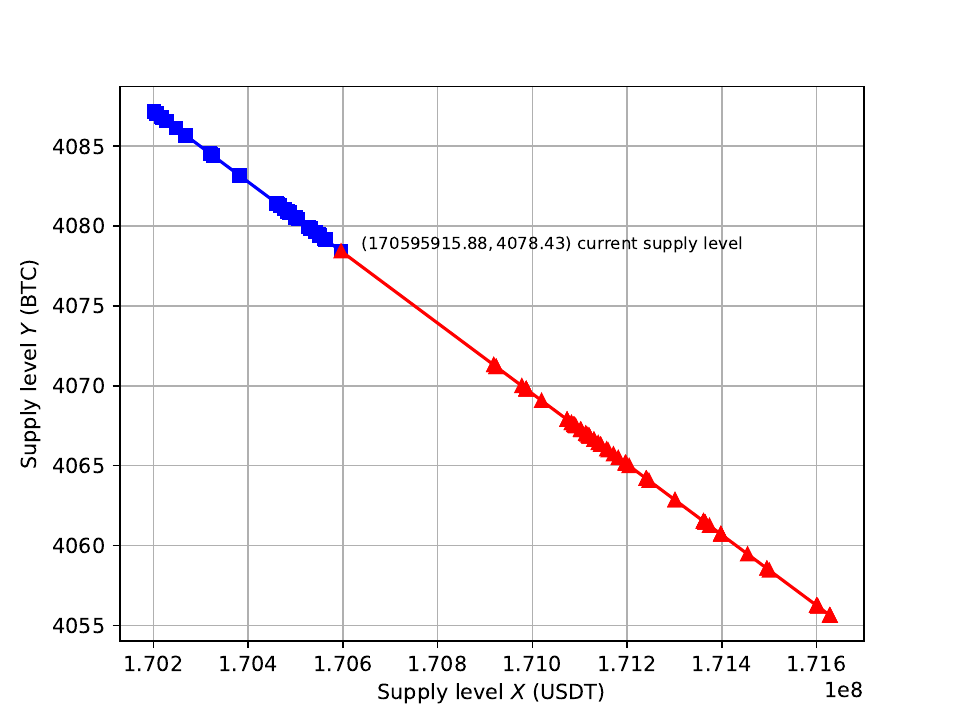} 
\vspace{-1em}
\includegraphics[scale=0.65]{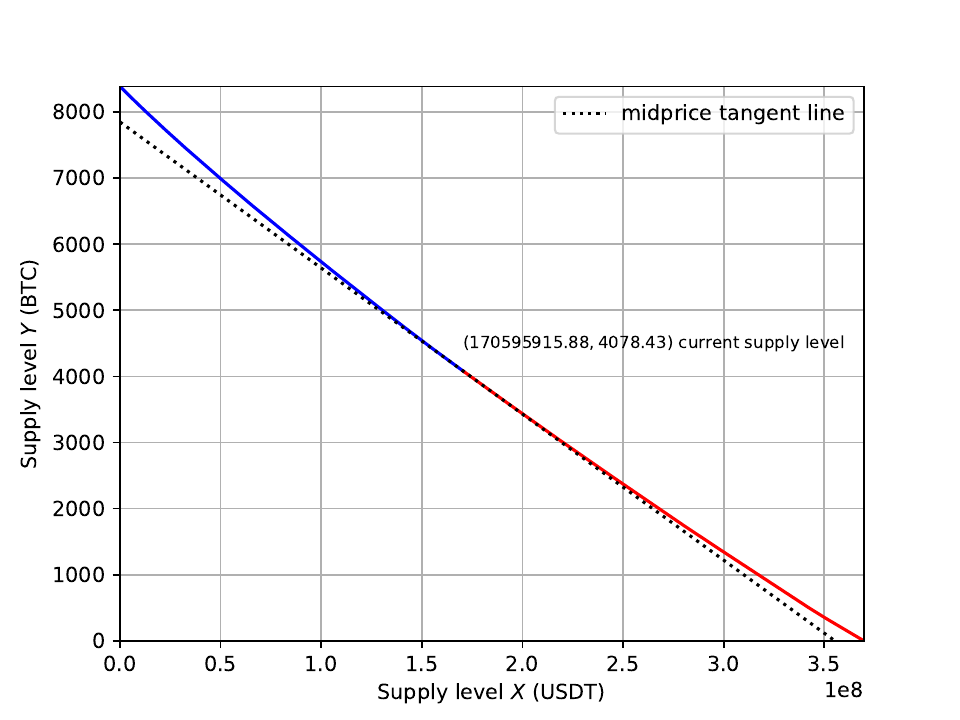} 
\caption{\emph{Top:} Iso-util of Binance's limit order book of the asset pair $(X,Y)=(\text{USDT}, \text{BTC})$ observed on January 2, 2024 at 05:21:16 a.m.~(UTC) for the first 100 levels on both sides. \emph{Bottom:} Iso-util of the entire limit order book. The tangent line at the current supply level $(x,y)$ has slope $-1/45265.815$ with midprice $45,265.815$.}
\label{fig:LOB_BTCUSDT_isoutil}
\end{figure}

Let us now turn to the general case. To this end, we recall the generalized inverse of a function.

\begin{definition}[Generalized inverse]\label{def:generalized_inverse}
For a non-decreasing function~$f:(0, \infty) \to \mathbb{R}$ the generalized inverse~$f^{-1}: \mathbb{R} \to [0, \infty)$ is defined by
\begin{align}
f^{-1} (y) := \sup \left\{ x \in (0, \infty) \ | \ f(x) \leq y \right\},
\end{align}  
where we use the convention~$\sup \emptyset = 0$.
For a non-increasing function~$f:(0, \infty) \to \mathbb{R}$ the generalized inverse~$f^{-1}: \mathbb{R} \to (0, \infty)$ is defined by
\begin{align}
f^{-1} (y) := \sup \left\{ x \in (0, \infty) \ | \ f(x) > y \right\},
\end{align}  
where we again use the convention~$\sup \emptyset =0$.
\end{definition}

\begin{lemma}\label{prop:properties_generalized_inverse}
The generalized inverse defined in Definition~\ref{def:generalized_inverse} satisfies the following properties:
\begin{enumerate}
\item[(i)] \label{point:generalized_inverse_strictly_increasing} The generalized inverse~$f^{-1}$ of a non-decreasing (non-increasing) function is non-decreasing (non-increasing).
\item[(ii)] \label{point:generalized_inverse_right_continuous} The generalized inverse~$f^{-1}$ is right-continuous.
\item[(iii)] \label{point:generalized_inverse_proper_inverse} If the function $f$ is strictly increasing or strictly decreasing, then the generalized inverse coincides with the proper inverse.
\end{enumerate}
\end{lemma}

The proof of the previous lemma is standard and hence omitted. See, for example, Proposition 1 in~\cite{EmHo:13} for a similar statement for a slightly different definition of the generalized inverse.

\begin{proposition}[Associating a limit order book to an iso-util] \label{prop:isoutil_to_lob}
Consider an iso-util~$I$ with current supply level~$(x_0,y_0)$. Let~$f$ denote the function representation of~$I$ from Proposition~\ref{p_iso_util_function}, and let~$f'^{-1}$ denote the generalized inverse of its derivative. Then for~$p \in (0 , \infty)$ it holds that
\begin{align}
F_s(p) := \max \left\{  y_0 - f(f'^{-1} (-p^{-1})) , 0 \right\}
\end{align}
defines the remaining supply function and
\begin{align}
    F_d(p) := \max \left\{   f(f'^{-1} (-p^{-1})) - y_0 , 0 \right\}
    \end{align}
defines the remaining demand function of a settled limit order book. 
\end{proposition}

\begin{proof}[Proof of Proposition~\ref{prop:isoutil_to_lob}]
To show that~$F_d$ is a remaining demand function we need to verify the conditions of Definition~\ref{def:remaining_demand_supple_function}, namely that~$F_d$ is non-increasing, right-continuous and satisfies~$\lim_{p \to \infty} F_d (p) = 0$. 

We start with showing that the remaining demand function~$F_d$ is non-increasing and right-continuous. We observe that the function~$i:p \mapsto - p^{-1}$ is strictly increasing. Because the function~$f$ is convex, the derivative~$f'$ is non-decreasing. Therefore, by Lemma~\ref{prop:properties_generalized_inverse} the generalized inverse~$f'^{-1}$ is also non-decreasing. We also observe that the function~$f$ is non-increasing. Hence the combination~$f \circ f'^{-1} \circ i$ is also non-increasing. This yields that the remaining demand function~$F_d(p)$ is non-increasing.

Next, we verify that the remaining demand function~$F_d$ is right-continuous. We observe that the function~$i: p \mapsto - p^{-1}$ is strictly increasing. Also, by Lemma~\ref{prop:properties_generalized_inverse} the generalized inverse is right-continuous. Additionally, as we have seen above, the generalized inverse~$f'^{-1}$ is non-decreasing. Because the function~$f$ is continuous and non-increasing, this implies by verifying the definition that the combination~$f \circ f'^{-1}$ is right-continuous, which yields that the remaining demand function~$F_d(p)$ is also right-continuous. 

Let us now show that~$\lim_{p \to \infty} F_d (p) = 0$. We start with observing that the function~$F_d (p)$ is non-increasing and non-negative, i.e., $F_d(p) \geq 0$. Hence, it suffices to show that~$F_d(p) = 0$ for some~$p$. We recall that~$f'_{-}$ denotes the left-hand derivative of~$f$. Let us choose~$ p_d =  - \frac{1}{f'_{-} (x_0)}$. Then, it holds by definition of the generalized inverse that~$f'^{-1}(-p_d^{-1}) \geq x_0 $. Hence, because the function~$f$ is decreasing it follows that
\begin{align}
f(f'^{-1}(-p_d^{-1})) \leq f(x_0) = y_0,
\end{align}
which implies the desired identity~$F_d(p_d) =0$.

The argument that the remaining supply function~$F_s$ is non-decreasing and right-continuous is similar and left out. To show that~$\lim_{p \to 0} F_s (p) = 0$, we observe that~$F_s (p)$ is non-decreasing and non-negative. Hence, it suffices to find a value~$p_s$ such that~$F_s(p_s) =0$. We recall that~$f'_{+}$ denotes the right-hand derivative of~$f$. Choosing~$p_s=  - \frac{1}{f'_{+} (x_0)}$ we observe that~$f'^{-1}(-p_s^{-1}) \leq x_0$. Now, a similar argument as for~$F_d$ implies the desired identity~$F_s(p_s) =0$.

The last step is to show that the limit order book given by~$F_d$ and~$F_s$ is settled. By~\eqref{eq:ordering_left_right_hand_derivative2} we have $p_d \leq p_s$, which implies the desired conclusion by observing that~$F_d(p_d)= F_s(p_s)=0.$ 
\end{proof}

\begin{figure}
\includegraphics[scale=0.7]{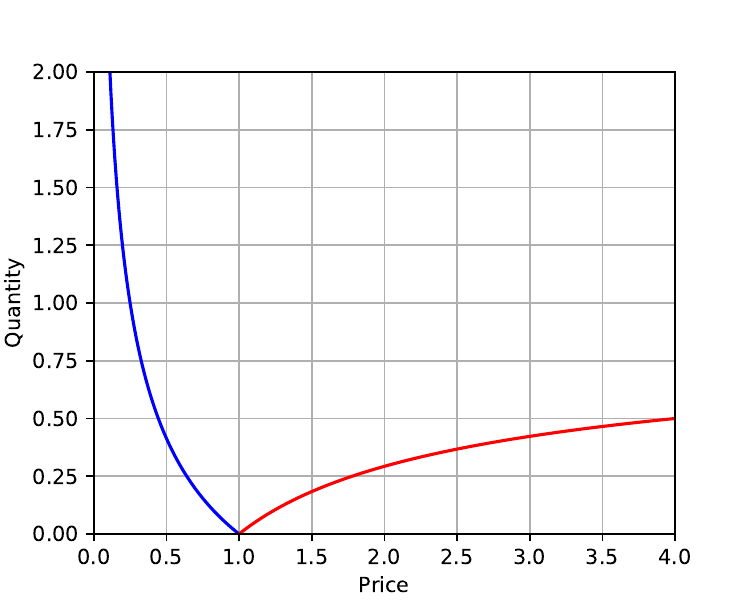}
\caption{Remaining demand function (blue) and remaining supply function (red) for the iso-util $x \cdot y =1$ with current supply level $(1,1)$.}
\label{fig:rsf_rdf_ideal_market}
\end{figure}

\begin{example} \label{ex:rsf_rdf_ideal_market}
Let us examine an ideal market from Definition~\ref{def:ideal_market}, which is given by the utility function $U(x,y)= \log x +  \log y$ and used by Uniswap v2 as its automated market maker protocol. We consider the iso-util $x \cdot y =T=A^2$ at temperature $T$ or, equivalently, with mean activity $A$. Let $(x_0,y_0)$ denote the current supply level. Then, due to Proposition~\ref{p_gradient_characterization_prices}, the marginal price is $p_0 = \frac{x_0}{y_0}$ and a straightforward application of Proposition~\ref{prop:isoutil_to_lob} shows that the remaining demand function (RDF) is given by
\begin{align}
F_d(p) = \begin{cases}
A \left( \frac{1}{\sqrt{p}} - \frac{1}{\sqrt{p_0}}\right) ,  & \text{for } p < p_0 , \\
0, & \text{for } p \geq p_0.
\end{cases}
\end{align}
Therefore, the demand measure $\mu_d$ is given by the Lebesgue density
\begin{align}
f_d(p) = \begin{cases}
\frac{A}{2p^{\frac{3}{2}}}, & \text{for } p < p_0, \\
0, & \text{for } p \geq p_0.
\end{cases}
\end{align} 
Moreover, the remaining supply function (RSF) is given by
\begin{align}
F_s(p) = \begin{cases}
0, & \text{for } p < p_0 , \\
A \left( \frac{1}{\sqrt{p_0}} - \frac{1}{\sqrt{p}},
 \right) & \text{for } p \geq p_0,
\end{cases}
\end{align}
and the supply measure $\mu_s$ has the Lebesgue density
\begin{align}
f_s(p) = \begin{cases}
0 & \text{for } p < p_0 , \\
\frac{A}{2p^{\frac{3}{2}}} & \text{for } p \geq p_0.
\end{cases}
\end{align}
We refer to Figure~\ref{fig:rsf_rdf_ideal_market} for an illustration. From this calculation, we see that in an ideal market temperature modulates the available liquidity; e.g., in a market four times as ``hot'', there is twice as much liquidity.
\end{example}

Let us now turn to associating an iso-util to a limit order book. Again, we need some preparation.

\begin{definition}[Pricing function]
Consider a limit order book $\mathscr{L}$ given by the remaining demand function~$F_d$ and the remaining supply function~$F_s$. Then the bid pricing function~$p_b$ is defined as~$p_b:=F_d^{-1}$ and the ask pricing function~$p_a$ is defined as~$p_a:=F_s^{-1}$. 
\end{definition}

\begin{remark}
The pricing functions~$p_b$ and~$p_a$ have a simple interpretation: If one wants to buy or sell~$y$ many units of~$Y$, then the price of the~$y$-th unit is given by~$p_a(y)=F_s^{-1}(y)$ and~$p_b(y) = F_d^{-1}(y)$, respectively.
\end{remark}

The following statement is an immediate consequence of involved definitions.

\begin{lemma}\label{p_strict_positivity_of_pricing_functions}
Assume that the remaining demand function~$F_d$ and the remaining supply function~$F_s$ are not the zero functions. Then the pricing functions~$p_b$ and~$p_a$ are strictly positive for~$y>0$, i.e., they satisfy~$p_b(y)>0$ and $p_a(y) >0$ for all~$y>0$. 
\end{lemma}

\begin{definition}[Depth of a limit order book]
The depth of the bid side of the limit order book is defined as
\begin{align}
d_b(y) := \int_0^y p_b(t) \ dt.
\end{align}
The depth of the ask side of the limit order book is defined as
\begin{align}
d_a(y) := \int_0^y p_a(t) \ dt.
\end{align} 
\end{definition}

\begin{remark}[Interpretation of depth]
The depth $d_a(y)$ of the ask side is the amount of money needed to buy $y$ many units, and the depth $d_b(y)$ of the bid side is the amount of money received from selling $y$ many units. 
\end{remark}

We also have the elementary observation which follows from Lemma~\ref{p_strict_positivity_of_pricing_functions}.
\begin{lemma}
The bid depth~$d_b$ is strictly increasing on the interval~$[0,\lim_{p \to 0 } F_d(p))$ and $d_b(y) = \infty$ for any~$y > \lim_{p \to 0} F_d(p)$. Similarly, the ask depth~$d_a$ is strictly increasing on the interval $[0,\lim_{p \to \infty } F_s(p))$ and ~$d_a(y) = \infty$ for any~$y > \lim_{p \to \infty} F_s(p)$.
\end{lemma}

When translating a limit order book into an iso-util we need to determine the current supply levels~$(x_0,y_0)$, which will only be well-defined if the limit order book is bounded. 

\begin{definition}[Bounded limit order book]
We say that a limit order book $\mathscr{L}$ given by the remaining supply function~$F_s$ and remaining demand function~$F_d$ is bounded if the following limits exists:
\begin{align}\label{eq:finite_limit_order_book_x_0}
x_0 := \int_{0}^\infty p \ d F_d(p)< \infty  
\end{align} 
and
\begin{align}\label{eq:finite_limit_order_book_y_0}
y_0 : = \lim_{p \to \infty} F_s(p) < \infty.
\end{align} 
Here, the integral denotes the Lebesgue-Stieltjes integral.
\end{definition}

\begin{proposition}[Associating an iso-util to a limit order book] \label{prop:lob_to_isoutil}
We consider a bounded limit order book given by a remaining demand function~$F_d$ and a remaining supply function~$F_s$. We define the function~$y_b: (0, x_0] \to [0, \infty)$ by
\begin{align}
y_b (x) :=  y_0 + d_b^{-1}(x_0-x)
\end{align}
and the function~$y_a : [x_0, \infty) \to [0, \infty)$ by
\begin{align}
y_a (x) := \max \left\{ y_0 - d_a^{-1} (x-x_0), 0 \right\}.
\end{align}
Then~$y_b$ defines the bid part of an iso-util~$I_b$ and~$y_a$ defines the ask part of an iso-util~$I_a$. If the limit order book is settled, then the function~$y: (0, \infty) \to [0, \infty)$ given by
\begin{align}
y(x) := 
\begin{cases}
y_b(x) & \mbox{if } x \leq x_0 \\
y_a(x) & \mbox{if } x > x_0
\end{cases}
\end{align}
is convex and its graph~$I = I_b \cup I_a$ defines an iso-util with current supply level~$(x_0, y_0)$, bid part~$I_b$ and ask part~$I_a$. 
\end{proposition}

\begin{proof}[Proof of Proposition~\ref{prop:lob_to_isoutil}]
We need to show that the functions~$y_b$ and~$y_a$ are convex, that~$\lim_{x \to \infty} y_a(x)=0$, and that in the settled case the function~$y(x)$ is convex (cf.~Remark~\ref{rem:charcaterization_function_representation}). Let us first show that the function~$y_a$ is convex. It suffices to show that its derivative~$y_a'$ is non-decreasing. Straightforward differentiation yields for~$x>x_0$
\begin{align}
y_a'(x) = - \frac{1}{d_a'(d_a (x-x_0))} = - \frac{1}{p_a(d_a (x-x_0))} = - \frac{1}{F_s^{-1}(d_a (x-x_0))}.
\end{align}
We observe that~$F_s$ is non-decreasing and therefore also~$F_s^{-1}$ (see Lemma~\ref{prop:properties_generalized_inverse}). Because~$d_a$ is increasing (see Lemma~\ref{p_strict_positivity_of_pricing_functions}) it follows that~$p_a(d_a (x-x_0))$ is non-decreasing, which in turn implies that~$y_a'(x)$ is non-decreasing.

The argument to show that the function~$y_b$ is convex is similar and left out. Also, the property~$\lim_{x \to \infty} y_a(x)=0$ follows directly from the definitions.

It is left to show that if the limit order book is settled then the function~$y$ is convex. For this let us first note that the function~$y$ is continuous, which follows from the observation~$y_a(x_0) = y_b(x_0)$. To show that~$y$ is convex, it suffices to show that~
\begin{align}\label{eq:iso_util_to_lob_settled}
{y_b}'_{-}(x_0) \leq {y_a}'_{+}(x_0).
\end{align}
The last inequality means that the left-hand derivative of~$y_b$ at~$x_0$ is less than or equal to the right hand derivative of~$y_a$ at~$x_0$. Let~$y_m$ denote a mid price of the limit order book. Because the limit order book is settled it follows from the definitions that for any~$h>0$
\begin{align}
d_{b}^{-1} (h) \leq y_m \leq d_{a}^{-1} (h). 
\end{align}
Straightforward manipulation of the last inequality using the definition yields 
\begin{align}
    \frac{y_0 - y_b(x_0 -h )}{h} \leq \frac{y_a(x+h)-y_0}{h}, 
\end{align}
which implies the desired estimate~\eqref{eq:iso_util_to_lob_settled} by sending~$h \to 0$.
\end{proof}

\begin{definition} [Iso-utils of a limit order book]\label{def:unsettled_isotuil}
   Let us consider the union~$I=I_b\cup I_a $ of the bid part~$I_b$ and ask part~$I_a$ of a limit order book (see Proposition~\ref{prop:lob_to_isoutil}). By a slight misuse of terminology we also call~$I$ an iso-util even if it may be non-convex and thus cannot be an iso-util of a utility function.
\end{definition}
    
In contrast to iso-utils of a utility function $U$, iso-utils of a limit order book might be non-convex, which at the same time identifies an arbitrage opportunity.

\begin{proposition}[Arbitrage and non-convex iso-utils]\label{prop:arbitrage_non_convex} 
The iso-util of a limit order book is non-convex if and only if the limit order book is unsettled. If that is the case then there is an arbitrage opportunity in the market.      
\end{proposition}

\begin{proof}[Proof of Proposition~\ref{prop:arbitrage_non_convex}]
If the limit order book is unsettled, there is a bid limit order (buy order) with a higher price than an ask limit order (sell order). Therefore, the left hand derivative of the iso-util at the current supply level is strictly smaller than its right hand derivative, which contradicts convexity. The arbitrage opportunity is straightforward, buying the asset from the ask order and selling it for a better price to the bid order.    
\end{proof}

\begin{remark}[Absence of arbitrage in Definition~\ref{d_markets} of a market]\label{rem:quasi_concave_no_arbitrage}
Let us consider a utility function~$U$ of a market~$\mathcal{M}$. One of the defining conditions of a market is that the utility function~$U$ is quasi-concave, which means that the iso-utils~$I$ must be convex. Hence, quasi-concavity implies that there is no direct price arbitrage opportunity in the market.     
\end{remark}

\begin{example}[Non-convex iso-util of an unsettled limit order book]      
By virtue of Proposition~\ref{prop:arbitrage_non_convex}, if the limit order book is unsettled then the resulting graph of $I_b\cup I_a$ is non-convex. Only after settling matching orders, i.e., after the clearing procedure outlined in Proposition~\ref{prop:settling_a_market} (adiabatic) or Proposition~\ref{prop:iso_util_settling_a_market} (iso-util), one obtains a proper convex iso-util. Recall that Table~\ref{tab:discrete_LOBs} (left) shows an example of an unsettled limit order book. The ask part $I_a$ and bid part $I_b$ of the associated iso-util are illustrated in Figure~\ref{fig:LOB_isoutil_unsettled}. The graph of $I_b \cup I_a$ is not convex. The resulting convex iso-utils from the adiabtic clearing and iso-util clearing are also depicted in Figure~\ref{fig:LOB_isoutil_unsettled}. Recall that the corresponding limit order books are given in Table~\ref{tab:discrete_LOBs} (right)
and Table~\ref{tab:example_adiabatic_isoutil} (right). Observe that both clearing procedures give rise to a new method of convexifying a graph.
\end{example}

\begin{figure}
\centering
\includegraphics[scale=0.65]{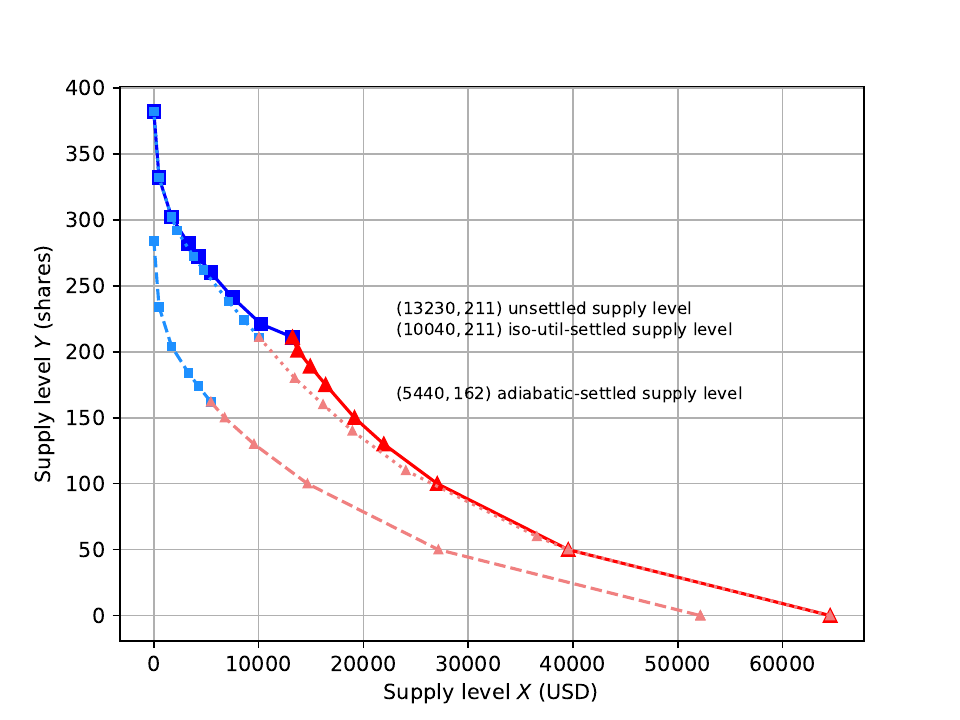} 
\caption{Comparison of the unsettled non-convex iso-util (solid line) generated from the unsettled limit order book given in Table~\ref{tab:discrete_LOBs} (left) with the convex iso-utils associated to the iso-util clearing (dotted line) and adiabatic clearing (dashed line), which is also illustrated in Figure~\ref{fig:LOB_isoutil_settled}.}
\label{fig:LOB_isoutil_unsettled}
\end{figure}

\section{Aggregation of markets}\label{sec:aggregate_markets}

In this section, we describe the arbitrage-mediated aggregation mechanism of financial markets. As discussed above, we consider two different aggregation processes: adiabatic and iso-util aggregation. In Definition~\ref{d_aggregate_market} below we outline the details of the aggregation mechanisms. In Example~\ref{ex_aggregation_ideal_markets} and~\ref{ex_aggregation_ideal_markets_different_prices} we calculate the adiabatic and iso-util aggregation of two ideal markets. Remarkably, it turns out that the iso-util aggregation of two ideal markets is again an ideal market. We then discuss applications of market aggregation, the fundamental law of market dynamics, and the role of market dynamical entropy. 

\begin{definition}[Adiabatic and iso-util market aggregation] \label{d_aggregate_market}
For a pair of assets~$(X,Y)$, we consider $n \in \mathbb{N}$ markets $\mathcal{M}_j$, $j \in \{1,\ldots,n\}$, with associated iso-utils $I_j$ and supply levels $(x_j, y_j)$. We denote by $F_{s,j}$ and $F_{d,j}$ the corresponding remaining supply and remaining demand function of the iso-util $I_j$.
\begin{enumerate}
\item[(i)] The unsettled aggregated limit order book is given by the remaining supply function $F_{s,u} = \sum_{j=1}^n F_{s,j}$ and remaining demand function $F_{d,u} = \sum_{j=1}^n F_{d,j}$. The unsettled aggregated iso-util $I_u$ has supply level $(x_u, y_u) = (\sum_{j=1}^n x_j, \sum_{j=1}^n y_j)$ and is given by the iso-util that is associated to the limit order book $(F_{s,u} , F_{d,u})$ (see Proposition~\ref{prop:lob_to_isoutil}).
\item[(ii)] In the case of adiabatic aggregation, the settled aggregated limit order book $(F_{s,a}, F_{d,a})$ is given by the adiabatic clearing of the unsettled limit order book $(F_{s,u}, F_{d,u})$ following Proposition~\ref{prop:settling_a_market}; and in the case of iso-util aggregation the iso-util-settled limit order book $(F_{s,i}, F_{d,i})$ is obtained from the iso-util clearing of the unsettled limit order book $(F_{s,u}, F_{d,u})$ described in Proposition~\ref{prop:iso_util_settling_a_market}. The associated iso-utils are given by $I_a$ and $I_i$, respectively, with corresponding supply levels $(x_a,y_a)$ and $(x_i,y_i)$ as described in Proposition~\ref{p_adiabatic_iso_util_ebntropy}.
\end{enumerate}
\end{definition}

\begin{notation}
The operation of unsettled aggregation of markets is denoted with the symbol $\circledtriangle$ whereas $\bigtriangleup$ indicates the operation of settled aggregation (adiabatic or iso-util). More precisely, we denote with
\begin{align}
\mathcal{M}_1 \circledtriangle \mathcal{M}_2, \qquad U_1 \circledtriangle U_2, \qquad  I_1 \circledtriangle I_2
\end{align}
the unsettled aggregation of the markets $\mathcal{M}_1$ and $\mathcal{M}_2$, utility functions $U_1$ and $U_2$, and iso-utils $I_1$ and $I_2$; and with 
\begin{align}
\mathcal{M}_1 \bigtriangleup \mathcal{M}_2, \qquad U_1 \bigtriangleup U_2, \qquad  I_1 \bigtriangleup I_2
\end{align}
we denote the settled aggregated counterparts.
\end{notation}
    
\begin{example}[Aggregation of two ideal markets with same marginal price]\label{ex_aggregation_ideal_markets}
Let us consider two ideal markets $\mathcal{M}_1$ and $\mathcal{M}_2$ with utility function $U(x,y) = \log x + \log y$ (recall Definition~\ref{def:ideal_market}). For the first market, we consider an iso-util $x \cdot y = T_1 = A_1^2$ with temperature $T_1$, mean activity $A_1$, and current supply level $(x_1,y_1)$. For the second market, we consider an iso-util $x \cdot y = T_2= A_2^2$ with temperature $T_2$, mean activity $A_2$, and current supply level $(x_2,y_2)$. We assume that the marginal prices of both markets coincide, which implies by Proposition~\ref{p_gradient_characterization_prices} that $p_1 = \frac{x_1}{y_1} = p_2 = \frac{x_2}{y_2}$. Since the marginal prices are the same, there is no overlap of the corresponding remaining demand and remaining supply functions, and adding them up readily yields the remaining demand function and remaining supply function of the aggregated market which is settled. In particular, we have $\mathcal{M}_1 \circledtriangle \mathcal{M}_2 = \mathcal{M}_1 \triangle \mathcal{M}_2$, there are no price arbitrage opportunities (i.e., $P=0$ in~\eqref{e_arbitrage_profit}), no entropy is generated from aggregation ($S_i=S_a=(0,0)$ in Definition~\ref{d_adiabatic_iso_util_entropy}) and the aggregated market has supply level $(x_u,y_u) = (x_i,y_i) = (x_a,y_a) = (x_1+x_2, y_1+y_2)$ (cf.~Proposition~\ref{p_adiabatic_iso_util_ebntropy}). Setting $p_e := p_1 = p_2$, Example~\ref{ex:rsf_rdf_ideal_market} yields that the remaining demand function (RDF) of the aggregated market is given by
\begin{align}
F_d(p) = 
\begin{cases}
(A_1 + A_2) \left( \frac{1}{\sqrt{p}} - \frac{1}{\sqrt{p_e}}\right), & p < p_e , \\
0, & p \geq p_e,
\end{cases}
\end{align}
and the remaining supply function (RSF) of the aggregated market is given by
\begin{align}
F_s(p) = \begin{cases}
0, & p < p_e, \\
(A_1+A_2) \left( \frac{1}{\sqrt{p_e}} - \frac{1}{\sqrt{p}} 
 \right), & p \geq p_e.
\end{cases}
\end{align}
Since RDF and RSF characterize iso-utils, we can conclude that the aggregated market $\mathcal{M}_1 \triangle \mathcal{M}_2$ is again an ideal market with utility function $U(x,y) = \log x +  \log y$ at temperature $\left( \sqrt{T_1}+\sqrt{T_2} \right)^2$ and mean activity $A_1+A_2$. Hence, when aggregating ideal markets with same marginal prices, the mean activity is additive and the temperature is super-additive.
\end{example}

\begin{example}[Aggregation of two ideal markets with different marginal prices]\label{ex_aggregation_ideal_markets_different_prices}
We consider the same setup as in Example~\ref{ex_aggregation_ideal_markets}. This time, however, we assume that the marginal prices of the two markets $\mathcal{M}_1$ and $\mathcal{M}_2$ are different, that is, $p_1= \frac{x_1}{y_1} <  p_2= \frac{x_2}{y_2}$.
\begin{itemize}
\item[(i)] The unsettled aggregated market $\mathcal{M}_1 \circledtriangle \mathcal{M}_2$ has remaining demand function ($\text{RDF}_{u}$)
\begin{align}
F_{d,u}(p) = \begin{cases}
(A_1+A_2) \frac{1}{\sqrt{p}} -  \left( \frac{A_1}{\sqrt{p_1}} + \frac{A_2}{\sqrt{p_2}} \right), & p < p_1 , \\
A_2 \left( \frac{1}{\sqrt{p}} - \frac{1}{\sqrt{ p_2}}\right), & p_1 \leq p < p_2, \\
0, & p \geq  p_2,
\end{cases}
\end{align}
and remaining supply function ($\text{RSF}_{u}$)
\begin{align}
F_{s,u}(p) = 
\begin{cases}
0, & p <  p_1 , \\
A_1 \left( \frac{1}{\sqrt{ p_1}} - \frac{1}{\sqrt{p}} 
\right), & p_1 \leq p < p_2, \\
\left( \frac{A_1}{\sqrt{p_1}} + \frac{A_2}{\sqrt{p_2}} \right) - (A_1 + A_2) \frac{1}{\sqrt{p}}, & p_2 \leq p.
\end{cases}
\end{align}
The corresponding unsettled supply level $(x_u,y_u)$ in~\eqref{p_adiabatic_iso_util_ebntrop_eq} is given by
\begin{equation}
x_u = A_1 \sqrt{p_1} + A_2 \sqrt{p_2}  \quad \text{and} \quad y_u = \frac{A_1}{\sqrt{p_1}} + \frac{A_2}{\sqrt{p_2}}.
\end{equation}
\item[(ii)] For the clearing mechanism, we can compute the clearing volume $Z$ in~\eqref{equ:clearing_volume_2} with unique clearing price $p_e$ in~\eqref{def:auction_price} (i.e., $p_e = p_s = p_d$ in Lemma~\ref{prop:clearing_auxiliary_lemma}) as 
\begin{align}
Z = & \, A_2 \left( \frac{1}{\sqrt{p_e}} - \frac{1}{\sqrt{p_2}} \right) = A_1 \left( \frac{1}{\sqrt{p_1}} - \frac{1}{\sqrt{p_e}} \right), \\ p_e = & \, \frac{p_1p_2(A_1+A_2)^2}{(A_1\sqrt{p_2} + A_2 \sqrt{p_1})^2}.
\end{align}
Using Proposition~\ref{prop:settling_a_market}, the adiabatic-settled aggregated market $\mathcal{M}_1 \triangle \mathcal{M}_2$ has remaining demand function ($\text{RDF}_{a}$)
\begin{align}
F_{d,a}(p) = \begin{cases}
(A_1+A_2) \frac{1}{\sqrt{p}} - \left( \frac{A_1}{\sqrt{p_1}} + \frac{A_2}{\sqrt{p_2}} \right) - Z, & p < p_1 , \\
A_2 \left( \frac{1}{\sqrt{p}} - \frac{1}{\sqrt{ p_2}}\right) - Z, & p_1 \leq p < p_e, \\
0, & p \geq  p_e,
\end{cases}
\end{align}
and remaining supply function ($\text{RSF}_{a}$)
\begin{align}
F_{s,a}(p) = 
\begin{cases}
0, & p < p_e , \\
A_1 \left( \frac{1}{\sqrt{ p_1}} - \frac{1}{\sqrt{p}} 
\right) -  Z, & p_e \leq p < p_2, \\
\left( \frac{A_1}{\sqrt{p_1}} + \frac{A_2}{\sqrt{p_2}} \right) - (A_1 + A_2) \frac{1}{\sqrt{p}} - Z, & p_2 \leq p.
\end{cases}
\end{align}
The resulting arbitrage profit $P$ from aggregation defined in~\eqref{e_arbitrage_profit} is given by 
\begin{equation}
P = A_2 (\sqrt{p_2} - \sqrt{p_e}) - A_1 (\sqrt{p_e} - \sqrt{p_1}),
\end{equation}
and the resulting adiabtic entropy in~\eqref{d_adiabatic_iso_util_entropy_s_a} amounts to
\begin{equation}
S_a = \left( A_2 (\sqrt{p_2} - \sqrt{p_e}), A_1 \left( \frac{1}{\sqrt{p_1}} - \frac{1}{\sqrt{p_e}} \right) \right).
\end{equation}
This is precisely the loss of liquidity of asset $X$ and $Y$ due to the arbitrage-mediated adiabatic aggregation. Specifically, according to~\eqref{p_adiabatic_iso_util_ebntropy_a}, the new supply level after adiabatic aggregation is
\begin{equation}
(x_a,y_a) = (x_u,y_u) - S_a = \left( A_1 \sqrt{p_1} + A_2 \sqrt{p_e}, \frac{A_1}{\sqrt{p_e}} + \frac{A_2}{\sqrt{p_2}} \right).
\end{equation}
Note that the adiabatic-aggregated market $\mathcal{M}_1 \triangle \mathcal{M}_2$ is not an ideal market. The associated iso-util can be computed via Proposition~\ref{prop:lob_to_isoutil}.
\item[(iii)] Following Proposition~\ref{prop:iso_util_settling_a_market}, the remaining demand function ($\text{RDF}_{i}$) and remaining supply function ($\text{RSF}_{i}$) of the iso-util-settled aggregated market $\mathcal{M}_1 \triangle \mathcal{M}_2$ can be computed as 
\begin{align}
F_{d,i}(p) = 
\begin{cases}
(A_1 + A_2) \left( \frac{1}{\sqrt{p}} - \frac{1}{\sqrt{p_e}}\right), & p < p_e , \\
0, & p \geq p_e,
\end{cases}
\end{align}
and
\begin{align}
F_{s,i}(p) = \begin{cases}
0, & p < p_e, \\
(A_1+A_2) \left( \frac{1}{\sqrt{p_e}} - \frac{1}{\sqrt{p}} 
 \right), & p \geq p_e.
\end{cases}
\end{align}
The resulting iso-util entropy in~\eqref{d_adiabatic_iso_util_entropy_s_i} is given by
\begin{equation}
S_i = \left( A_2 (\sqrt{p_2} - \sqrt{p_e}) - A_1 (\sqrt{p_e} - \sqrt{p_1}), 0 \right), 
\end{equation}
and the new supply level in~\eqref{p_adiabatic_iso_util_ebntropy_i} after iso-util aggregation is
\begin{equation}
(x_i,y_i) = (x_u,y_u) - S_i = \left( (A_1 + A_2) \sqrt{p_e}, \frac{A_1+A_2}{\sqrt{p_e}} \right),
\end{equation}
where we use the identity $\frac{A_1}{\sqrt{p_1}} + \frac{A_2}{\sqrt{p_2}} = \frac{A_1+A_2}{\sqrt{p_e}}$. Observe that due to the (arbitrage-mediated) iso-util aggregation, there is actually no loss in liquidity in asset $Y$, only in asset $X$. Remarkably, similar to Example~\ref{ex_aggregation_ideal_markets}, we observe that the iso-util-aggregated market $\mathcal{M}_1 \triangle \mathcal{M}_2$ is again an ideal market with utility function $U(x,y) = \log x +  \log y$ at temperature $\left( \sqrt{T_1}+\sqrt{T_2} \right)^2$ and mean activity $A_1+A_2$. In contrast to Example~\ref{ex_aggregation_ideal_markets}, the entropy generated from iso-util aggregation is not zero, and a quantity of $A_2 (\sqrt{p_2} - \sqrt{p_e}) - A_1 (\sqrt{p_e} - \sqrt{p_1})$ of asset $X$ is lost due to arbitrage. Also note that the clearing price $p_e$ from the aggregation indeed coincides with the new marginal price at the new supply level $(x_i,y_i)$ since $\frac{x_i}{y_i} = p_e$. 
\end{itemize}  
\end{example}

\begin{remark}
From a practical perspective, the aggregation of financial markets allows to reconcile multiple individual markets into one global market. This is particularly useful in settings where markets are fragmented. For example, in the crypto market there are hundreds of centralized and decentralized exchanges trading the same asset pairs, e.g., BTC and USDT. The same applies for the forex market. Aggregation of those markets gives a better picture of the true price, supply and demand relation. Specifically, observe that Example~\ref{ex_aggregation_ideal_markets_different_prices} explicitly shows that the iso-util aggregation of two Uniswap v2 liquidity pools of the same asset pair but with different marginal prices again results in a liquidity pool with Uniswap v2 protocol. Also recall from Remark~\ref{rem:isoutil_liquidprov} that since liquidity providers in these pools act iso-util, iso-util aggregation is the correct notion for the aggregation of these two pools.
\end{remark}

\begin{remark}
One theoretical implication of our proposed aggregation mechanisms of financial markets is that utility functions unify both dual interpretations of a market: on the one hand as a mechanism to exchange assets, and on the other hand as a place where traders interact and come to an agreement. In Section~\ref{s_axiomatic_introduction_markets}, we explain how utility functions (and iso-utils in particular) can be used to describe the possible trades on a market, modeling the market as an exchange mechanism. Market aggregation explains how the utility function of a market arises from the utility functions of the individual traders. For this recall that a trader with her utility function can be regarded as an ``atomic'' market. We make the assumption that all traders are transparent in the sense of Definition~\ref{d_transparent_tradent}. We also assume that they are isolated, that is, they cannot change their initial portfolio/individual supply levels $(x_i,y_i)$ before aggregation (for instance, by interacting with a market) and they cannot communicate their preferences to each other. Additionally, we assume that arbitrageurs are always present, able, and willing to trade on arising arbitrage opportunities. After aggregation, the new portfolios/supply levels of the individual traders will be Pareto-optimal. The corresponding marginal prices of the individual traders will be consistent and coincide in the smooth case.
\end{remark}

When aggregating markets with different marginal prices there will be an overlap of buy and sell limit orders in the joint limit order book. Hence, the aggregated market will be unsettled and will have a non-convex iso-util. After canceling the overlapping buy and sell orders out of the limit order book one gets a \emph{settled} aggregated market with a convex iso-util. We discussed the details on the settling procedure in Section~\ref{s_isoutil_LOB}.\\

When transitioning from the unsettled to the settled market the limit order book gets cleared. Both clearing mechanisms, adiabatic and iso-util, result in a negative supply change as liquidity is leaving the market and lost to arbitrage. Using adiabatic clearing the lost liquidity is given by the adiabatic entropy $S_a$; using iso-util clearing it is given by the iso-util entropy $S_i$. This yields the following law.

\begin{theorem}[Fundamental Law of Market Dynamics]
\emph{When markets aggregate some liquidity is lost to arbitrage.}
\end{theorem}

This law shares similarities with the Second Law of Thermodynamics which states that in a heat engine, i.e., a machine that transforms heat into work, not all heat is transformed into work but some is transformed into entropy. The energy transformed into entropy is \emph{lost} in the sense that it cannot be used to generate work anymore. This is the reason why we choose the terminology adiabatic and iso-util entropy to describe the lost liquidity in market aggregation. In practice market dynamical entropy may be used to quantify the arbitrage profits for brokers and market makers using for example dark pools to match buyers and sellers. It can also be used as a measure of price uncertainty, and might play a role to quantify volatility.\medskip

It is possible to apply this aggregation mechanism to other markets. In Appendix~\ref{sec:entropy_in_economics} we highlight some aspects of applying the theory to consumer and producer markets in economics. One could say, in an exaggerated manner, that market-dynamical entropy is the driving force behind economic activity. 

\begin{remark}[Aggregating markets with different assets]\label{r_aggregating_assets_different_assets} We only consider the aggregation of markets of the same underlying asset pair~$(X,Y)$. It would be an interesting problem to investigate the aggregation of markets that only share one asset. For example, how to join a market of the asset pair $(X_1, X_2)$ with the market of the asset pair $(X_1, X_3)$. The main idea would be to use the pathway via the limit order book correspondence and settling.    
\end{remark}

\section{The emerging theory of market dynamics and open problems}\label{sec:market_dynamics}

Let us briefly describe the main inspiration and goals of the emerging theory of market dynamics. Starting point is the observation that a thermodynamic system, e.g., a piston, can be reduced to its functionality: It is a mechanism to exchange volume for pressure and vice versa. This shares a lot with the functionalist approach to markets which sees a market as a mechanism to exchange one asset into another. Hence, it might not be surprising that meta principles and ideas from thermodynamics are very useful when studying the structure of markets and the interaction between markets. One of the main goals of thermodynamics is to describe how and why energy transfers occur if thermodynamic systems are brought into contact. We propose that the main goal of the theory of market dynamics should be to describe how markets and traders interact if brought into contact. In this article, we made the first step toward achieving this goal: We describe how markets aggregate in the simplistic framework of a static, transparent market of two assets. \\

We introduced new -- and renamed some existing -- notions to point out the connection to thermodynamics. For example, to make precise the similarities to isotherms in thermodynamics, we call indifference curves \emph{iso-utils}. Isotherms denote curves in the state space that conserve the temperature/energy of the system. In market dynamics, iso-utils denote curves that preserve the utility of the portfolio/market. In thermodynamics, the term adiabatic refers to state changes that do not preserve the temperature. Therefore, we call a trader \emph{adiabatic} if her utility increased after a trade. Another example is the terminology of an \emph{ideal market} to describe a market with utility function $U(x,y) = \log x + \log y $. The reason is that the associated iso-utils have the same form as the isotherms of an \emph{ideal gas}. This association was also made in~\cite{MioriCucuringu:23} where the relation is called \emph{ideal crypto law}. \\

The \emph{law of market dynamics} and the term \emph{(market-dynamical) entropy} (see Section~\ref{sec:aggregate_markets}) are also inspired by thermodynamics. The second law of thermodynamics states that when energy flows from the hotter to the colder system, some heat energy is transformed into entropy. This energy is lost in the sense that it cannot be transformed into work anymore. Our fundamental law states that in market aggregation, some liquidity is lost to arbitrage. We call (market-dynamical) entropy the amount of liquidity that got lost through arbitrage. The liquidity is lost because the arbitrageur does not have the intention to re-inject her arbitrage profit back into the market. As we illustrate in Appendix~\ref{sec:entropy_in_economics}, entropy can also be used to measure the size of economic activity that resulted from bringing economic agents into contact with each other. \\

From the present status, there are many directions for future research. For example, one could study in more detail the ideal market, the role of temperature, the notions of \emph{(market-dynamical) work} and \emph{financial engines}. We believe that work might indicate how much value is transferred from one asset to another. Financial engines are inspired by heat engines and seem to be a very useful tool to describe financial bubbles, especially in the context of trading meme stocks, which are driven by bullish sentiments of retail investors; cf., e.g., \cite{SEC2021}. One could also study the role of market-dynamical entropy as another source of \emph{price fluctuations}. The main idea is to split up price volatility into two parts: The first part representing the iso-util activity of traders on the visible market, and the second part representing adiabatic activity caused by the entropy of the hidden market. \\

It might also be natural to ask if there are more concepts in market dynamics that correspond to thermodynamic objects (e.g., energy, free energy, internal energy, temperature, and enthalpy). If so, what relations and main principles do they satisfy? To give an example from thermodynamics, the formula $F = U - S$ connects the Helmholtz free energy $F$ to the internal energy $U$ and the entropy~$S$. Once the theory of market dynamics is sufficiently developed and its main principles are identified, one could turn to developing a theory of market mechanics. Like in statistical mechanics, the main goal of market mechanics would be to derive the main principles of market dynamics from the mechanics that govern the microscopic interaction of traders. \\

There are many ways to refine our model. To keep the presentation as simple as possible we concentrated on markets where only a pair of assets $(X_1,X_2)$ is traded. An extension to multiple assets~$(X_1, \ldots, X_n)$ would be quite beneficial. For example the joint market of the assets dollar, copper, and tin can describe the interrelationship between those assets a lot better than the separate markets of the asset pairs dollar/copper and dollar/tin. The reason is that copper and tin share a combined utility in creating the alloy bronze, a metal with superior properties compared to either element alone. The key to producing bronze lies in achieving the optimal ratio of copper to tin as this mixture offers the highest level of strength and durability. This ratio is close to 12\% tin and 88\% copper.  Consequently, the utility function for these two elements should reach its peak when supply levels of copper and tin are maintained at the precise ratio required for bronze production, underscoring their interdependent nature in the realm of metallurgy. Related, an intriguing question is to identify compatibility conditions under which two markets~$\mathcal{M}_1$ and~$\mathcal{M}_2$ with asset pairs~$(X_1, X_2)$ and~$(X_2, X_3)$ combine into a single joint market~$\mathcal{M}_3$ of the assets~$(X_1, X_2, X_3)$; or to study how markets of different assets are aggregated. \\

Another direction would be to study aggregation in non-transparent markets (see also Appendix~\ref{sec:hidden_markets}). This would be the area of~\emph{hidden market models}, which could be inspired by hidden Markov chains. The main question would be: Is it possible to learn the hidden part of the utility function of the market from data?\\

There are a lot of interesting research questions to address before the aggregation of financial markets is sufficiently understood. Then, the next natural but challenging step would be to study how utility functions of individual traders influence each other once they are brought into contact via the aggregated market.

\section*{Acknowledgment}
The authors want to extend their gratitude to Zachary Feinstein. His excellent presentation at the UCLA research seminar in Mathematical Finance and the follow-up discussions inspired the authors to this research project. Additionally, the first author would like to thank Stefano Olla for pointing out the importance of thermodynamics in the natural sciences. 

\bibliographystyle{alpha}
\bibliography{bib}

\appendix

\section{Aggregation of consumer and producer markets} \label{sec:entropy_in_economics}

In this appendix, we further illustrate adiabatic market aggregation with a simple example in the context of economics. We define the iso-utils of single consumers and producers, study their joint aggregated unsettled market, discuss possible obstacles for settling, and finally consider their adiabatic-settled aggregated market; see Example~\ref{ex_consumer_market}--\ref{ex_aggregated_consumer_producer} below. As a consequence, we see that some economic activities like consumption, production, and trade can be explained as a consequence of market aggregation and resulting entropy. 

\begin{definition}[Consumer and producer market]
A state of a market is called consumer market if the current supply level $(x,y)$ satisfies $y=0$. As a consequence the associated iso-util only has a bid part. Similarly, a state of a market is called producer market if the current supply level~$(x,y)$ satisfies~$x=0$. As a consequence the associated iso-util only has an ask part. 
\end{definition}

\begin{figure}
\centering
\includegraphics[scale=0.6]{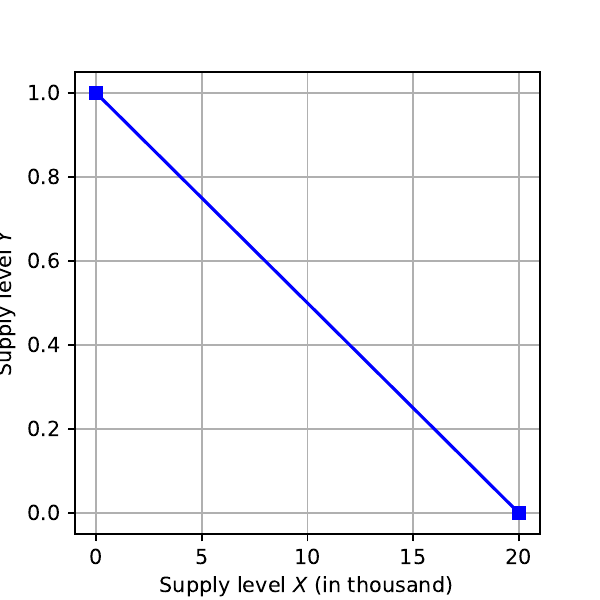}
\includegraphics[scale=0.6]{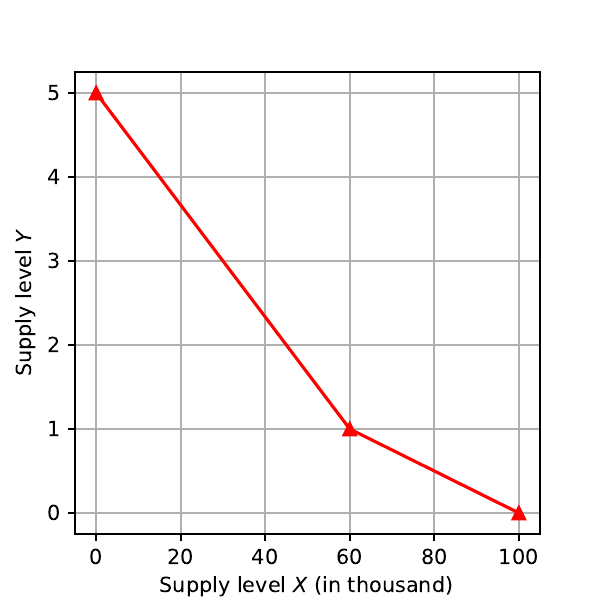}
\caption{\emph{Left:} Iso-util of a single consumer with no car who is willing to buy at most one car (no use for a second one) for at most the price of \$20,000 (not more money). Squares mark admissible supply levels, and the current supply level is given by $(20000,0)$. \emph{Right:} Iso-util of a single car producer that is able to produce a single car at a price of \$40,000 but needs to produce at least 4 cars at a price of~\$15,000. Triangles mark admissible supply levels, and the current supply level is given by $(0,5)$.}
\label{fig:isoutil_single_consumer_producer}
\end{figure}

\begin{example}[Consumer market] \label{ex_consumer_market}
Figure~\ref{fig:isoutil_single_consumer_producer} (left) gives an example of a consumer market. We consider the asset pair $(X,Y)= (\text{dollar}, \text{cars})$. The figure illustrates an hypothetical iso-util of one individual consumer who is willing to buy at most one car (no need or space for a second) for the price of at most \$20000 (not more money available). Hence, the current supply level is given by $(20000, 0)$. In this example the supply levels $y$ are discrete, i.e., the set of admissible supply levels is given by $\mathcal{P} = \left\{ (x, y ) \ | \ x >0 \text{ and } y \in \left\{0,1 \right\} \right\}$ (see also Remark~\ref{r_discrete_supply_levels}).
\end{example}

\begin{example}[Producer market]\label{ex_producer_market}
Figure~\ref{fig:isoutil_single_consumer_producer} (right) gives an example of a producer market. Again, the asset pair $(X,Y)= (\text{dollar} , \text{cars} )$ is considered. The figure illustrates an hypothetical iso-util of one individual producer. The producer is able to produce one car at the price of \$40000. By using economies of scale the producer is also able to produce at least four cars at a price of \$15000 each. The current supply level is given by $(0, 5)$. The supply levels $y$ are discrete, i.e., the set of admissible supply levels is given by $\mathcal{P} = \left\{ (x, y ) \ | \ x >0 \text{ and } y \in \left\{0,1,5 \right\}\right\}$ (see also Remark~\ref{r_discrete_supply_levels}).
\end{example}

\begin{figure}
\centering
\centering
\includegraphics[scale=0.6]{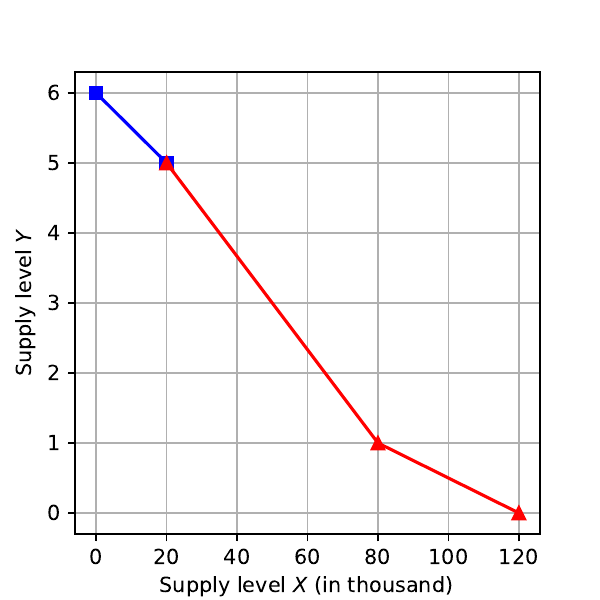}
\includegraphics[scale=0.6]{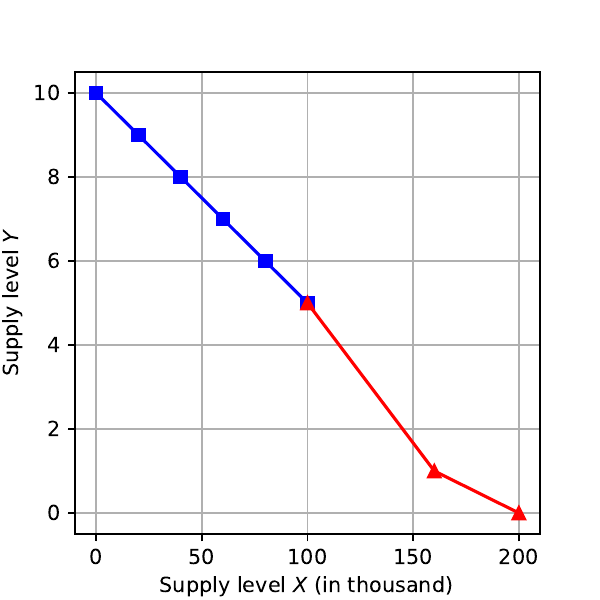}
\caption{\emph{Left:} Iso-util of the aggregated market of the consumer and the car producer of Figure~\ref{fig:isoutil_single_consumer_producer}. We observe that the iso-util is non-convex and therefore the market is not settled. The market cannot settle because the consumer is not willing to pay \$40,000 for a car and the producer cannot produce one car for less than \$20,000. This means that the admissible supply levels (given by the squares and triangles) are not compatible. The current supply level in this figure is given by $(20000,5)$. \emph{Right:} The market from the left plot aggregated with four more consumers with iso-utils as in Figure~\ref{fig:isoutil_single_consumer_producer} (left). In this case, the market can settle as the admissible supply levels of the ask and bid part are compatible. Here, settling means that the producer will produce four cars, sell them to a car dealer (trader) who resells them to the four consumers. We refer to Figure~\ref{fig:isoutil_single_producer_consumer_aggregated_settled} for the associated settled iso-util which will be convex.}
\label{fig:isoutil_single_producer_consumers}
\end{figure}

\begin{example}[Unsettled aggregated consumer producer market]\label{ex_aggregated_consumer_producer}
We aggregate the consumer market of Example~\ref{ex_consumer_market} and the producer market of Example~\ref{ex_producer_market} and obtain an unsettled iso-util given by Figure~\ref{fig:isoutil_single_producer_consumers} (left). Because the producer can only produce at least four cars for a competitive price and the consumer is not able to pay more for one car and cannot purchase four cars, the admissible supply levels of the bid part are not compatible with the admissible supply levels of the ask part. Therefore, the market cannot settle. If three more consumers with the same iso-util as in Example~\ref{ex_consumer_market} can be found, then the admissible supply levels are compatible and the market could settle. In practice, this means that a car dealer (trader) would take advantage of the situation, buy the four cars from the producer and sell them to individual consumers. Hence, a trade can be understood as a consequence of market aggregation in combination with the second law of market dynamics. We refer to Figure~\ref{fig:isoutil_single_producer_consumers} (right) and Figure~\ref{fig:isoutil_single_producer_consumer_aggregated_settled} for more details. Therein, we aggregated the market with four instead of three more consumers to get a better visualization.
\end{example}

\begin{figure}
\centering
\includegraphics[scale=0.6]{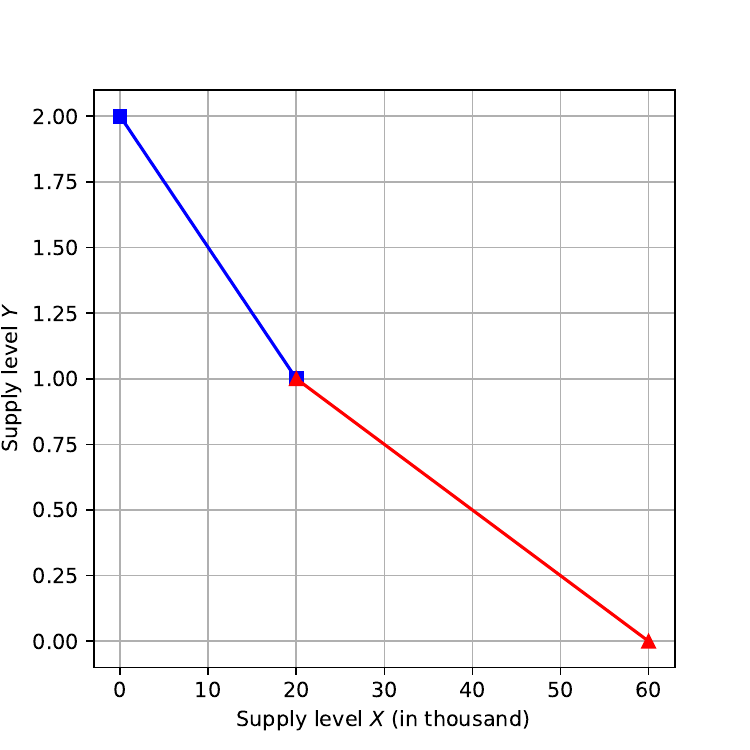}
\caption{Settled iso-util of the unsettled iso-util of Figure~\ref{fig:isoutil_single_producer_consumers} (right). }
\label{fig:isoutil_single_producer_consumer_aggregated_settled}
\end{figure}

From the discussion of the Figures~\ref{fig:isoutil_single_consumer_producer}--\ref{fig:isoutil_single_producer_consumer_aggregated_settled} it becomes obvious that the framework of market dynamics is capable to model complex economic actions from a micro-economic to a macro-economic level. Using the multi-asset definition of markets, i.e., allowing a finite number of assets $(X_1, X_2, \ldots, X_n)$, one could imagine to construct from ground up the world market of all assets. \\

In summary, we have illustrated how entropy can be used to measure the size of economic activity. This opens an interesting perspective to policy makers: Maximizing the market-dynamical entropy corresponds to maximizing economic activity. How can this be achieved? By making economic agents interact as much as possible. As entropy is maximized in transparent markets, the policy maker should set rules which reduce obstacles to interaction as much as possible, and promote and reward transparency such that economic agents reveal their iso-utils. However, just concentrating on achieving transparency is not necessarily the best, as transparency obviously stands in competition with other policy goals like privacy protection.

\section{Hidden markets}\label{sec:hidden_markets}

When applying the arbitrage-mediated aggregation mechanism to real financial markets one faces many challenges. First, it is very hard to observe markets as they are highly fragmented: in decentralized finance by design and in traditional finance by evolution; cf.~\cite{SEC2020}. For instance, the equity market in the US is comprised of several exchanges, broker-dealer platforms, alternative trading systems, and dark pools; and most of them are not directly observable. Second, as we have described in the introduction, traders typically do not act transparent. When aggregating markets, one needs to deal with \emph{hidden information} and \emph{hidden liquidity} (see, e.g., Section 1.4 in~\cite{Bouchaud:18}). \\

There is even a bigger problem than the fact that traders do not act transparent: The utility function of an individual might not be known to the individual herself.\footnote{This might explain the immense value of data gathering via smartphones, social media, web browsing, etc. Actions reveal preferences, which might even be unknown to the individual herself. Through those actions, hidden utility functions can be learned and then monetized.} Many retail traders rely on their intuition, sentiments, etc. The utility function of an individual  trader might also be heavily influenced by the utility function of others, as everyone looks at the overall market in order to derive the value of their assets. \\

To account for these problems, one approach is to split up the utility function in a visible and a hidden part. Aggregating would then yield two markets:  
\begin{itemize}
\item An \emph{aggregated hidden market} which is comprised of the aggregation of the complete utility functions. This market is not observable and unsettled.    
\item An \emph{aggregated visible market} which is comprised of the aggregation of the visible part of the utility functions. This market is observable and settles via arbitrage. 
\end{itemize} 

As a consequence, the hidden aggregated market contains \emph{hidden arbitrage opportunities}. Estimating the hidden aggregated market is a subtle but rewarding problem as it might uncover those opportunities. This might explain the long-term success of certain players in financial markets. Sentiment analysis on social media, which was used very successfully to generate trading signals, can be interpreted as an attempt to estimate the hidden part of the utility function of a typical retail trader. This perspective also sheds a new light on price discovery as the hidden utility function becomes visible and settles around the newly discovered price. \\

Hidden arbitrage also challenges the Efficient Market Hypothesis, which argues that all available information is already incorporated into the asset price. Therefore, fluctuations would only result from new information entering the market. After aggregation, only visible but no hidden information is incorporated in the present price. As some information is hidden, the correct price is not known to the market participants. The overlap of the hidden limit order book yields an additional source of price fluctuations as it means that traders do not agree on the correct price. Only if \emph{all} market participants are able to learn the hidden information faster than it changes, then this source of fluctuation vanishes. Otherwise, one will observe a market with large, persistent, entropic-driven price fluctuation. \\

\end{document}